\documentclass[10pt]{article}

\usepackage{algorithm}
\usepackage{algorithmic}
\usepackage{bbm}
\renewcommand{\algorithmicrequire}{ \textbf{Input:}} 

\usepackage{amsmath}
\usepackage{multirow}

\usepackage{graphicx}
\usepackage{epstopdf}
\usepackage{extarrows}
\usepackage{float}
\usepackage[latin1]{inputenc}
\usepackage{tikz}
\usepackage{mathrsfs}
\usepackage{amssymb}
\usepackage{stmaryrd}
\usepackage{mathabx}
\usepackage[normalem]{ulem}
\usepackage{flushend}
\usepackage{amsfonts}
\usetikzlibrary{trees,arrows,automata}
\usepackage{tikz-qtree}
\usepackage{multirow}
\usepackage{tikz}

\usepackage{ntheorem}

\hfuzz=\maxdimen
 \usepackage{indentfirst}
\usepackage{subfigure}
\usepackage{array}
\newcolumntype{C}[1]{>{\centering\arraybackslash}p{#1}}
\usepackage[pdftex, bookmarksnumbered, bookmarksopen, colorlinks, citecolor=blue, linkcolor=blue]{hyperref}

\newtheorem{definition}{Definition}
\newtheorem{theorem}{Theorem}
\newtheorem{lemma}{Lemma}
\newtheorem{corollary}{Corollary}
\newtheorem{remark}{Remark}
\newtheorem{example}{Example}
\newtheorem{proposition}{Proposition}

\newtheorem*{proof}{Proof}

\usepackage[justification=centering]{caption}
\usepackage{caption}
\begin{document}

\title{Supervisory Control of Quantum Discrete Event Systems
}

\author{Daowen Qiu
    \thanks{Institute of Quantum Computing and Computer Theory,  School of Computer Science and Engineering, Sun Yat-sen
     University, Guangzhou, 510006, China (e-mail: issqdw@mail.sysu.edu.cn).   }
}

\date{ }
\maketitle


{\bf Abstract.}
 Discrete event systems (DES) have been  deeply developed and applied in practice, but state complexity in DES still is an important problem to be better solved with innovative methods. With the development of quantum computing and quantum control,  a natural problem is to simulate DES by means of quantum computing models and to establish {\it quantum DES} (QDES).  The motivation is twofold:  on the one hand, QDES have potential applications when DES are simulated and processed by quantum computers, where quantum systems are employed to simulate the evolution of states driven by discrete events,
 and on the other hand, QDES may have essential advantages over DES concerning state complexity for imitating some practical problems.
So, the goal of this paper is  to establish a basic framework of QDES by using {\it quantum finite automata} (QFA)  as the modelling formalisms, and the supervisory control theorems of QDES are established and proved. Then we present a polynomial-time algorithm to decide whether or not the controllability condition holds. In particular, we construct a number of new examples of QFA to illustrate the supervisory control of QDES and to verify the essential advantages of QDES over classical DES in state complexity.


\vskip 2mm
\noindent


\vskip 2mm
\noindent


\vskip 2mm



\section{Introduction}

Discrete Event Systems (DES)  and Continuous-Variable Dynamic Systems (CVDS) are two important  classes of control systems \cite{book1}. Roughly speaking, the goal of systems to be controlled is to achieve some desired specifications, and feedback control means  using any available information from the system's
behavior  to adjust the control's input \cite{book1}. CVDES are time-varying dynamic systems and their state transitions are time-driven, but  the state transitions of DES are event-driven.
In general, the study of CVDS relies on differential-equation-based models, and DES  are simulated usually by automata and Petri nets \cite{book1,book2}.

 More exactly, DES are formally dynamical systems whose states are discrete and the evolutions of its states are driven by the occurrence of events \cite{RW87,book1,book2}. As Kornyak mentioned \cite{Kornyak13}, the study of discrete systems is also important from the practical point of view since many physical objects are in fact discrete. 
As a precise model in logic level, DES have been applied to many real-world systems, such as traffic systems, manufacturing systems, smart grids systems, database management systems, communication protocols, and logistic (service) systems, etc. However, for some practical systems modeled by large-scale states \cite{book2}, the complexity of processing systems still needs to be solved appropriately.

Supervisory Control Theory (SCT) of DES  is a basic and important subject in DES \cite{RW87,book1,book2}, and it was originally proposed by Ramadge and Wonham \cite{RW87,book1,book2}. A DES and the control specification are modeled as automata. 
The task of supervisors is to ensure that the supervised (or closed-loop) system generates a certain language called specification language.

SCT of DES  exactly supports the formulation of various control problems of standard types, and it usually is automaton-based. Briefly, a DES is modeled as the generator (an automaton) of a formal language, and certain events (transitions) can be
disabled by an external controller. The idea is to construct this controller so that the
events it currently disables depend on the past behavior of the
 DES in a suitable way.

Automata form the most basic class of DES models \cite{RW87,book1,book2}. They are intuitive, easy to use,
amenable to composition operations, and amenable to analysis as well (in the finite-state
case). However, the (conventional) DES model cannot characterize the probability  of probabilistic systems and the possibility of fuzzy systems that
exist commonly in engineering field and the real-world problems with fuzziness, impreciseness, and
subjectivity. So, probabilistic DES and fuzzy DES were proposed \cite{Lin90, LW1993, PPL2009,   LY02,Qiu05,DYQ19}.
However, to my best knowledgement, the state complexity problems still  have  not been studied in  probabilistic and fuzzy DES, and QFA have certain  advantages over {\it probabilistic finite automata} (PFA) in state complexity for some problems, so, with the development of quantum computing \cite{NC00},  how to establish {\it quantum DES} (QDES) and show certain advantages of state complexity  are a pending problem, and  this is the goal of the paper.



Quantum computers were first conceived by Benioff \cite{Benioff80} and Feynman \cite{Feynman82} in the early of 1980s, and in particular, Feynman \cite{Feynman82} indicated it needs exponential time to simulate the evolution of quantum systems in classical computers but quantum computers can perform efficient simulation. In 1985, Deutsch \cite{Deutsch85} elaborated and formalized Benioff and Feynman's ideas  by defining the notion of quantum Turing machines, and proposed  Quantum Strong Church-Turing Thesis: \textit{A quantum Turing machine can efficiently simulate any realistic model of computation},  which is an extension of the traditional Strong  Church-Turing Thesis:\textit{A probabilistic Turing machine can efficiently simulate any realistic model of computation}. In a way,  this also inspires  to establish QDES since  it is  natural to develop  control systems from classical  models  to quantum ones. 

In fact, after Shor's  discovery \cite{Shor94} of a polynomial-time algorithm on quantum computers for prime factorization, quantum computation has become a very active research area in quantum physics, computer science, and quantum control
\cite{NJP09,WM10, Lloyd2000, NWCL2000}.
 The study of quantum control usually has been focused on time-varying systems, and coherent feedback control (i.e., feedback control using a fully quantum system) has been deeply investigated \cite{WM10, Lloyd2000, NWCL2000}. 

Classical automata consist of the most fundamental class of DES models \cite{book1,book2}, but they may lead to large-scale state spaces when modeling complex systems \cite{PC05}.
Though there are  strategies to attack the problem of large state spaces \cite{book2,PC05}, we still hope to discover new methods for solving the state complexity from a different point of view. In fact, quantum finite automata (QFA) can be employed as a powerful tool, since QFA have significant advantages over crisp finite automata concerning state complexity \cite{AF98,AY2021}. An excellent and comprehensive survey on QFA was presented by  Ambainis and Yakaryilmaz \cite{AY2021}. Moreover,  QFA have been studied in physical experiment  \cite{MP20,  PHYF2022, TFLZZ2019}, and in particular, recently Plachta et al. \cite{PHYF2022} demonstrated an experimental implementation of multiqubit QFA using the orbital angular momentum (OAM) of single photons, and  showed that a high-dimensional QFA can  outperform classical finite automata in terms of the required memory in a way.

Actually, QFA have been interestingly applied to interactive proof systems \cite{NY2009}, in which QFA are verifiers. In addition,  QFA have been used to simulate chemical reactions with certain advantages of time complexity \cite{BZ2020}.





Since QFA have better advantages over classical finite automata in state complexity \cite{AF98}, QDES likely can solve such problems with essential advantages of states complexity over DES. 
Therefore, the purpose of this paper is to initial the study of QDES. 
As classical DES are modeled by one-way (probabilistic or fuzzy) finite automata \cite{RW87,book1,book2,     Lin90, LW1993, PPL2009,  LY02,Qiu05,DYQ19},    we would employ one of one-way QFA (1QFA) to simulate QDES. 

QFA can be thought of as a theoretical model of quantum computers in which the memory is finite and described by a finite-dimensional state space \cite{AY2021, QLMG12,Gruska99,BK19}.     \textit{Measure-once} 1QFA (MO-1QFA) were initiated by Moore and Crutchfield \cite{MC00} and  \textit{measure-many } 1QFA (MM-1QFA) were studied first by Kondacs and Watrous \cite{KW97}, where ``1" means ``one-way", that is, the tape-head moves only from left side to right side.  MO-1QFA and MM-1QFA were deeply studied by Ambainis and Freivalds \cite{AF98}, Brodsky and Pippenger \cite{BP02}, and other authors (see, e.g.,  \cite{AY2021}). Then other 1QFA   were also proposed and studied by Ambainis {\it et al.},  Nayak,  Hirvensalo,  Yakaryilmaz and Say, Paschen, Ciamarra,  Bertoni {\it et al.}, Qiu and Mateus {\it et al.} as well other authors (e.g., the references in \cite{AY2021,QLMG12,BK19}). These 1QFA include \cite{AY2021}: Latvian-1QFA (La-1QFA), Nayak-1QFA (Na-1QFA), General-1QFA (G-1QFA), 1QFA with ancilla qubits (1QFA-A), fully 1QFA (Ci-1QFA), 1QFA with control languages (1QFA-CL), and \textit{one-way quantum finite automata together with classical states} (1QFAC), where G-1QFA, 1QFA-A, Ci-1QFA, 1QFA-CL, and 1QFAC can recognize all regular languages with bounded error.



Qiu {\it et al.} \cite{QLMG12, MQL12, QLZMG11, LQ08}
 studied the decision problems regarding equivalence of QFA and minimization of states of QFA, where the equivalence method will be utilized in this paper for checking a controllability condition. 





In general,  the languages for simulating DES (and QDES) should be prefix-closured and regular    \cite{RW87,book1,book2}, but we will prove that MO-1QFA  cannot  recognize with cut-point (i.e., unbounded error) any prefix-closured regular language. So, MO-1QFA are not suitable for simulating QDES.  In addition, it is more complete if the 1QFA  modeling QDES can recognize all regular languages and such 1QFA are relatively concise. Actually,  1QFAC proposed in \cite{QLMS15} is a hybrid of MO-1QFA and \textit{deterministic finite automata} (DFA), and both MO-1QFA and DFA are two special models of 1QFAC. 1QFAC can also recognize all regular languages and the computing procedure of 1QFAC are concise in a way. So, in this paper, we would like to use 1QFAC to model QDES.

\par
The remainder of the paper is organized as follows.
In Section \ref{SECQFA} we first introduce the basics of quantum computing and then present the definitions of QFA (MO-1QFA and 1QFAC; DFA and PFA are also recalled) and related properties, as well as we recall the decidability method of equivalence for 1QFAC. In Section \ref{SECQDES}, we first recollect the  supervisory control of DES (the language and automaton models of DES and related parallel composition operation), then we  present QDES and corresponding  supervisory control formalization of QDES (we define QDES by means of QFA, define quantum supervisors, and the supervisory control of QDES is formulated); parallel composition of QDES and related properties are also given.
 In Section \ref{SECQSCT}, we first prove that MO-1QFA can not recognize any prefix-closured language, but, as we know \cite{QLMS15}, 1QFAC  can do it; also  a number of new prefix-closured languages are constructed to show 1QFAC's state complexity advantages over PFA and DFA. Therefore 1QFAC  are employed to simulate QDES. Then we establish a number of supervisory control theorems of QDES, and in particular, the new examples are given to illustrate the supervisory control dynamics of QDES, and to verify the advantages of QDES over DES concerning state complexity. In Section \ref{SECDCC}, we give a method to determine the control condition of QDES. More specifically, the detailed polynomial-time algorithm for testing the existence of supervisors is provided. Finally in Section \ref{SECCR}, we summarize the main results we obtain and mention related  problems for further   developing QDES.

\section {One-way Quantum Finite Automata (1QFA)} \label{SECQFA}

In this section we serve to review the definitions of MO-1QFA and 1QFAC   together with related properties, and  the decidability method of equivalence for 1QFAC is recalled. 
In the interest of readability, we first recall some basics of quantum computing that we will use in the paper. For the details concerning quantum computing, we can refer to \cite{NC00}, and here we just
briefly introduce some notation to be used in this paper.


\subsection{ Some notation on  quantum computing}

Let $\mathbb{C}$ denote the set of all complex numbers, $\mathbb{R}$
 the set of all real numbers, and $\mathbb{C}^{n\times m}$
 the set of $n\times m$ matrices having entries in $\mathbb{C}$. Given two matrices $A\in \mathbb{C}^{n\times m}$ and $B\in\mathbb{C}^{p\times q}$, their {\it tensor product} is the $np\times
mq$ matrix, defined as \[A\otimes B=\begin{bmatrix}
  A_{11}B & \dots & A_{1m}B \\
  \vdots & \ddots & \vdots \\
  A_{n1}B &\dots & A_{nm}B \
\end{bmatrix}.\]
  We get $(A\otimes B)(C\otimes D)=AC\otimes BD$ if the operations $AC$ and $BD$ can be done in terms of multiplication of matrices.

  Matrix $M\in\mathbb{C}^{n\times n}$ is said to be {\it
unitary} if $MM^\dagger=M^\dagger M=I$, where
 $\dagger$ denotes conjugate-transpose operation. $M$ is said to be {\it
Hermitian} if $M=M^\dagger$. For $n$-dimensional row vector
$x=(x_1,\dots, x_n)$, its norm denoted by $||x||$ is defined as
$||x||=\big(\sum_{i=1}^n x_ix_i^{*}\big)^{\frac{1}{2}}$, where
symbol $*$ denotes conjugate operation.  Unitary matrices preserve
the norm, i.e., $||xM||=||x||$ for each  $x\in \mathbb{C}^{1\times
n}$ and unitary matrix $M\in\mathbb{C}^{n\times n}$.


Any quantum system can be described by a state of Hilbert space. More specifically,
for a
quantum system with a finite basic state set $Q=\{q_1,\dots,
q_n\}$, every basic state $q_i$ can be represented by an
$n$-dimensional row vector $\langle q_i|=(0\dots1\dots0)$ having
only 1 at the $i$th entry (where $\langle \cdot|$ is
Dirac notation, i.e., bra-ket notation). At any time, the state of this system
is a {\it superposition} of these basic states and can be
represented by a row vector $\langle \phi|=\sum_{i=1}^nc_i\langle
q_i|$ with $c_i\in\mathbb{C}$ and $\sum_{i=1}^n|c_i|^2=1$; $|\phi\rangle$ represents the  conjugate-transpose of $\langle \phi|$. So, the quantum system is described by the Hilbert space ${\cal H}_Q$ spanned by the base $\{|q_i\rangle: i=1,2,\dots,n\}$, i.e. ${\cal H}_Q=\text{span}\{| q_i\rangle: i=1,2,\dots,n\}$.

The evolution of quantum system's states complies with unitarity.  More exactly,  suppose the current state of system is $|\phi\rangle$. If it is acted on by some unitary matrix $M_1$,  then the system's state is changed to the new current state $M_1|\phi\rangle$;  the second unitary matrix, say $M_2$, is acted on  $M_1|\phi\rangle$, the state is further changed to  $M_2 M_1|\phi\rangle$.
So, after  a series of unitary matrices $M_1, M_2, \ldots, M_k$ are performed in sequence, the system's state becomes $M_kM_{k-1}\cdots M_1|\phi\rangle$.

If we want to get some information from a quantum system, then  a measurement is made on its current state. Here we consider {\it projective measurement} (i.e., von Neumann measurement). A projective measurement is described by an {\it observable}  that is  a Hermitian matrix ${\cal O}=c_1P_1+\dots +c_s P_s$, where $c_i$ is its eigenvalue and, $P_i$ is the projector onto the eigenspace corresponding to $c_i$. In this case, the projective measurement of ${\cal O}$ has result set $\{c_i\}$ and projector set $\{P_i\}$. For example, given state $|\psi\rangle$ is made by the measurement ${\cal O}$, then the probability of obtaining result $c_i$ is $\|P_i|\psi\rangle\|^2$ and the state $|\psi\rangle$ collapses to $\frac{P_i|\psi\rangle}{\|P_i|\psi\rangle\|}$.

\subsection{Definitions of 1QFA}

For non-empty set
$\Sigma$, by $\Sigma^{*}$ we mean the set of all finite length
strings over $\Sigma$, and $\Sigma^n$ denotes the set of all
strings over $\Sigma$ with length $n$. For $u\in \Sigma^{*}$,
$|u|$ is the length of $u$; for example, if
$u=x_{1}x_{2}\ldots x_{m}\in \Sigma^{*}$ where $x_{i}\in \Sigma$,
then $|u|=m$.  For set $S$,  $|S|$ denotes the cardinality of $S$. First we recall the definitions of {\it deterministic finite automata} (DFA) and probabilistic finite automata (PFA).

\subsubsection{DFA and PFA} \label{DFAPFA}

A DFA \cite{HU79} can be described by a five-tuple ${\cal A}=(Q,\Sigma,\delta,q_{0},Q_{a})$,
where $Q$ is the finite set of states; $\Sigma$ is a finite alphabet
of input; $\delta: Q\times \Sigma\rightarrow Q$ is the
transition function (In what follows, ${\cal P}(X)$
represents the power set of set $X$.); $q_{0}\in Q$ is the initial state; and
$Q_{a}\subseteq Q$ is called the set of accepting (or called as ``marked" in DES) states. Indeed,
transition function $\delta$ can be naturally extended to
$Q\times\Sigma^{*}$ in the following manner: For any $q\in Q$, any
$s\in\Sigma^{*}$ and $\sigma\in\Sigma$,
$
\delta(q,\epsilon)=\epsilon,\hskip 2mm {\rm and} \hskip 2mm
\delta(q,s\sigma)=\delta(\delta(q,s),\sigma).
$

The language recognized by ${\cal A}$ is as $\{w\in\Sigma^*: \delta(q_0,w)\in Q_a\}$.
We can depict it as Fig. \ref{fig:DFA}.


\begin{figure}[htbp]
\centering
\includegraphics[width=0.9\textwidth]{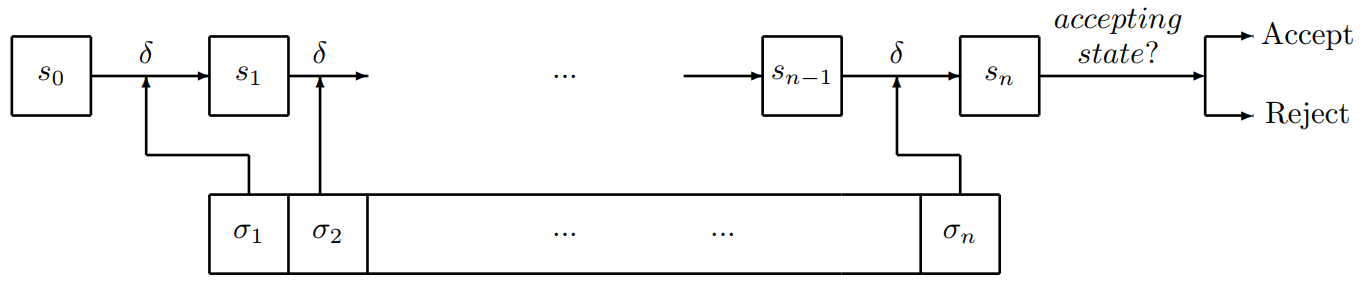}
\caption{The dynamics of DFA.}
\label{fig:DFA}
\end{figure}

A PFA \cite{Paz71} is also defined  as a five-tuple
 $\mathcal{M}=(S,\Sigma,\pi, \{M(\sigma)\}_{\sigma\in\Sigma},\eta)$, where
 $S=\{s_{1},s_{2},\cdots,s_{n}\}$ is a finite set of states; $\Sigma$ is a finite alphabet
of input;
 $\pi$ is an $n$-dimensional row vector, denoted as an initial distribution over $S$;
 $\forall \sigma\in\Sigma$, $M(\sigma)$ is an $n\times n$ random matrix (i.e., each row is a vector of transfer probabilities) and $M(\sigma)(i,j)$ denotes the probability of transferring from state $s_{i}$ to state $s_{j}$ after the machine $\mathcal{M}$ reads $\sigma$;
 $\eta=\begin{bmatrix}\eta_{1}, \cdots, \eta_{n} \end{bmatrix}^T$ is an $n$-dimensional column vector with elements $0$ or $1$, and $\eta_{i}=1$ means $s_{i}$ is an accepting state. For any input string $x=x_{1}x_{2}\cdots x_{n}\in \Sigma^{*}$, the probability that $\mathcal{M}$ accepts $x$ is 
\begin{equation}
P_{\mathcal{M}}(x)=\pi M(x_{1}) M(x_{2})\cdots M(x_{n})\eta.
\end{equation}




\subsubsection{MO-1QFA and 1QFAC}

MO-1QFA   are the simplest
quantum computing models proposed first by Moore and Crutchfield \cite{MC00}. In this model, the transformation on any
symbol in the input alphabet is realized by a unitary operator. A
unique measurement is performed at the end of a computation. More formally, an MO-1QFA with $n$ states and the input alphabet $\Sigma$ is a five-tuple \begin{equation}{\cal M}=(Q, |\psi_0\rangle,
\{U(\sigma)\}_{\sigma\in \Sigma}, Q_a,Q_r) ,\end{equation} where
$Q=\{|q_1\rangle,\dots,|q_n\rangle\}$ consist of  an orthonormal base that
spans a Hilbert space ${\cal H}_Q$;
 $|\psi_0\rangle\in {\cal H}$ is the initial state;
 for any $\sigma\in \Sigma$, $U(\sigma)\in \mathbb{C}^{n\times n}$ is a unitary matrix;
 $Q_a, Q_r\subseteq Q$ with $Q_a\cup Q_r=Q$ and $Q_a\cap Q_r=\emptyset$ are the accepting and rejecting states, respectively, and they describe an observable  by using the projectors
$P(a)=\sum_{|q_i\rangle\in Q_a}|q_i\rangle\langle q_i|$ and $P(r)=\sum_{|q_i\rangle\in Q_r}|q_i\rangle\langle q_i|$, with the result set $\{a,r \}$ of which `$a$'
and `$r$' denote ``accept'' and ``reject'', respectively. Here $Q$ consists of accepting and rejecting sets.

Given an MO-1QFA ${\cal M}$ and an input word $s=x_1\dots x_n\in\Sigma^{*}$, then starting from $|\psi_0\rangle$, $U(x_1),\dots, U(x_n)$ are applied in succession, and at the end of the word, a measurement $\{P(a),P(r)\}$ is performed with the result that ${\cal M}$ collapses into accepting states or rejecting states with
corresponding probabilities. Hence, the probability $L_{{\cal M}}(x_1\dots x_n)$ of ${\cal M}$ accepting $w$ is defined as:
\begin{align}
L_{\cal M}(x_1\dots x_n)=\|P(a)U_s|\psi_0\rangle\|^2\label{f_MO},
\end{align}
where we denote $U_s=U_{x_n}U_{x_{n-1}}\cdots U_{x_1}$. MO-1QFA can be depicted as Figure \ref{fig:mo1qfadyn}, in which if these unitary transformations are replaced by stochastic matrices and some stochastic vectors take place of
quantum states, then it is a PFA \cite{Paz71}.

\vskip 0mm

\begin{figure}[htbp]
\centering
\includegraphics[width=0.9\textwidth]{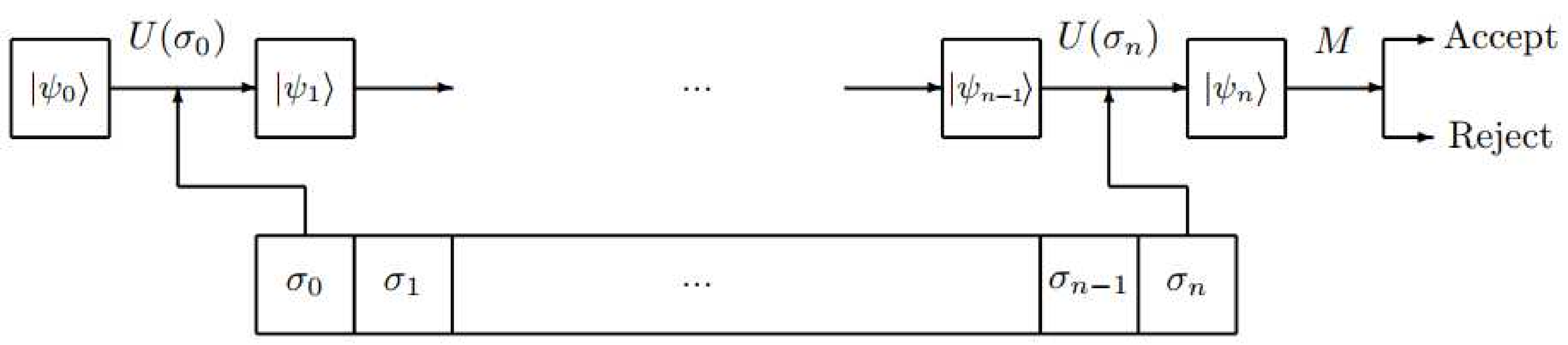}
\caption{MO-1QFA dynamics as an acceptor of languages.}
\label{fig:mo1qfadyn}
\end{figure}

Next we introduce 1QFAC proposed in \cite {QLMS15}.
A 1QFAC ${\cal A}$ \cite{QLMS15} is formally defined by a 8-tuple
\[
{\cal M}=(S,Q,\Sigma, s_{0},|\psi_{0}\rangle, \delta,
\mathbb{U}, {\cal P}),
\]
where:
\begin{itemize}
\item $\Sigma$  is a finite set (the {\it input alphabet});

\item $S$ is a finite set (the set of {\em classical states});

\item $Q$ is a finite set (the set of  {\em quantum  basis states});

\item $s_{0}$ is an element of $S$ (the {\em initial classical state});

\item $|\psi_{0}\rangle$ is a unit vector in the Hilbert space ${\cal H}(Q)$ (the {\em initial quantum state});

\item $\delta: S\times \Sigma\rightarrow S$ is a map (the {\em classical transition map});

\item $\mathbb{U}=\{U_{s\sigma}\}_{s\in S,\sigma\in \Sigma}$ where $U_{s\sigma}:{\cal H}(Q)\rightarrow {\cal H}(Q)$ is a unitary operator for each $s$ and $\sigma$ (the {\em quantum transition operator} at $s$ and $\sigma$);

\item ${\cal P}=\{{\cal P}_s\}_{s\in S}$  where each ${\cal P}_s$ is a projective measurement over ${\cal H}(Q)$ with outcomes  {\it accepting} (denoted by $a$) or {\it rejecting} (denoted by $r$)  (the {\em measurement operator at} $s$).
\end{itemize}

Hence, each ${\cal P}_s= \{P_{s,a}, P_{s,r}\}$ such that  $P_{s,a}+P_{s,r}=I$ and $P_{s,a}P_{s,r}= O$.
Furthermore, if the machine is in classical state $s$ and quantum state $|\psi\rangle$ after reading the input string, then  $\|P_{s,\gamma}|\psi\rangle\|^{2}$ is the probability of the machine producing outcome $\gamma$ on that input.

$\delta$ can be extended to a map $\delta^{*}:
\Sigma^{*}\rightarrow S$ as usual. That is,
$\delta^{*}(s,\epsilon)=s$; for any string $x\in\Sigma^{*}$ and
any $\sigma\in \Sigma$, $\delta^{*}(s,\sigma x)=
\delta^{*}(\delta(s,\sigma),x)$.
For the sake of convenience, we denote the map $\mu : \Sigma^{*}
\rightarrow S$, induced by $\delta$, as
$\mu(x)=\delta^{*}(s_{0},x)$ for any string $x\in\Sigma^{*}$.
We further describe the computing process of ${\cal
A}$ for input string
$x=\sigma_{1}\sigma_{2}\cdots\sigma_{n}$ where $\sigma_{i}\in
\Sigma$ for $i=1,2,\cdots,n$.

The machine ${\cal A}$ starts at the
initial classical state $s_{0}$ and initial quantum state
$|\psi_{0}\rangle$. On reading the first symbol $\sigma_{1}$ of the input string, the states of the machine change as follows: the classical state becomes
$\mu(\sigma_{1})$; the quantum state becomes $U_{s_{0}\sigma_{1}}|\psi_{0}\rangle$.
Afterward, on reading $\sigma_{2}$, the machine changes its classical state to $\mu(\sigma_{1}\sigma_{2})$ and its quantum state to the result of applying $U_{\mu(\sigma_{1})\sigma_{2}}$ to
$U_{s_{0}\sigma_{1}}|\psi_{0}\rangle$.

The process continues
similarly by reading
$\sigma_{3}$, $\sigma_{4}$, $\cdots$, $\sigma_{n}$ in succession.
Therefore, after reading $\sigma_{n}$, the classical state becomes
$\mu(x)$ and the quantum state is as follows:
\begin{align}
U_{\mu( \sigma_{1}\cdots \sigma_{n-2}\sigma_{n-1} )\sigma_{n}}U_{\mu(\sigma_{1}\cdots \sigma_{n-3}\sigma_{n-2})\sigma_{n-1}}\cdots
 U_{\mu(\sigma_{1})\sigma_{2}}
U_{s_{0}\sigma_{1}}|\psi_{0}\rangle.
\end{align}

Let ${\cal U}(Q)$ be the set of unitary operators on Hilbert space
${\cal H}(Q)$. For the sake of convenience, we denote the map
$v:\Sigma^{*}\rightarrow {\cal U}(Q)$ as: $v(\epsilon)=I$ and

\begin{equation}
v(x)=
U_{\mu( \sigma_{1}\cdots \sigma_{n-2}\sigma_{n-1} )\sigma_{n}}U_{\mu(\sigma_{1}\cdots \sigma_{n-3}\sigma_{n-2})\sigma_{n-1}}
\cdots U_{\mu(\sigma_{1})\sigma_{2}}
U_{s_{0}\sigma_{1}} \label{v}
\end{equation}
for $i=1,2,\cdots,n$, and $I$ denotes the identity
operator on ${\cal H}(Q)$, indicated as before.

By means of the denotations $\mu$ and $v$, for any input string
$x\in\Sigma^{*}$, after ${\cal A}$ reading $x$, the classical
state is $\mu(x)$ and the quantum state is $v(x)|\psi_{0}\rangle$.

Finally,   we obtain the
probability $L_{{\cal M}}(x)$ for accepting $x$:
\begin{equation}
L_{{\cal M}}(x)= \|P_{\mu(x),a}v(x)|\psi_{0}\rangle\|^{2}.
\end{equation}

1QFAC can be depicted as Figure \ref{fig:mm1qfadyn}.


\begin{figure}[htbp]
\centering
\includegraphics[width=0.9\textwidth]{ 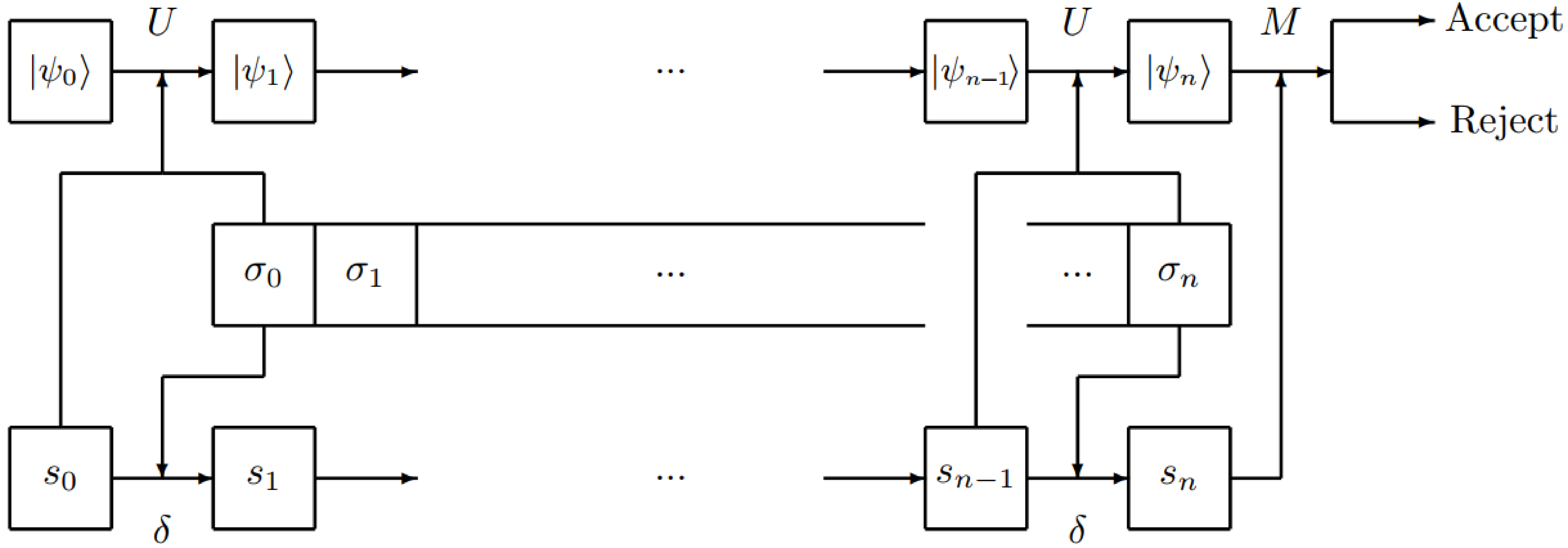}
\caption{1QFAC dynamics as an acceptor of languages.}
\label{fig:mm1qfadyn}
\end{figure}






\subsection{Determining the equivalence for  quantum finite automata}

In order to study the decision of controllability condition, in this subsection we introduce the method of how to determine the equivalence between 1QFAC, and the details are referred to \cite{LQ08, QLMS15, QLMG12}.

\begin{definition} \label{BLM}
A bilinear machine (BLM) over the alphabet $\Sigma$ is tuple
$\mathcal{M}=(S, \pi, \{M(\sigma)\}_{\sigma\in\Sigma},\eta)$, where $S$ is a finite state set with $|S|=n$, $\eta\in\mathbb{C}^{1\times n}$,   $\pi\in\mathbb{C}^{n\times 1}$,      and $M(\sigma)\in \mathbb{C}^{n\times n}$ for $\sigma\in\Sigma$.

\end{definition}

Associated to a BLM, the word function
\begin{equation}
L_{\mathcal{M}}: \Sigma^{\ast} \longrightarrow \mathbb{C}
\end{equation} is defined in the way: \begin{equation}L_{\mathcal{M}}(w)=\eta M(w_m)M(w_{m-1})\ldots M(w_1)\pi,\end{equation}  where $w=w_1w_2\ldots w_m\in \Sigma^{\ast}$.  In particular, when $L_{\cal M}(w)\in \mathbb{R}$ for
every $w\in \Sigma^{*}$,
 ${\cal M}$ is called a {\it real-valued bilinear  machine} (RBLM).

\begin{remark}\label{TPRBLM}
For any two RBLM \begin{equation} \mathcal{M}_i=(S_i, \pi_i, \{M_i(\sigma)\}_{\sigma\in\Sigma},\eta_i), i=1,2,\end{equation}  
it is easy to obtain  that, for any $w\in\Sigma^*$,
\begin{equation}
L_{\mathcal{M}_1\otimes \mathcal{M}_2}(w)=L_{\mathcal{M}_1}(w) \times L_{\mathcal{M}_2}(w);
\end{equation}
\begin{equation}
L_{\mathcal{M}_1\oplus \mathcal{M}_2}(w)=L_{\mathcal{M}_1}(w) + L_{\mathcal{M}_2}(w),
\end{equation}
where $\mathcal{M}_1\otimes \mathcal{M}_2$ and $\mathcal{M}_1\oplus \mathcal{M}_2$ are defined as usual \cite{MC00}. 
\end{remark}


\begin{definition}
Two BLM (RBLM, QFA) $\mathcal{M}_1$  and
$\mathcal{M}_2$ over the same alphabet $\Sigma$ are said to be
equivalent (resp. $k$-equivalent) if $L_{{\cal M}_1}(w)=L_{{\cal
M}_2}(w)$ for any $w\in \Sigma^{*} $ (resp. for any input string
$w$ with $ |w|\leq k$).
\end{definition}


The following proposition determines the equivalence between BLM (RBLM) \cite{CMR2006, KMOWW2011, LF2015, WLY2021}.

\begin{proposition} \label{BLMEQ}
Two BLM (RBLM) $\mathcal{M}_1$ and
$\mathcal{M}_2$ with $n_1$ and $n_2$ states, respectively, are
equivalent if and only if they are $(n_1+n_2-1)$-equivalent.
Furthermore, there exists a polynomial-time algorithm running in
time $O((n_1+n_2)^3)$ that takes as input two BLM  (RBLM) $\mathcal{M}_1$
and $\mathcal{M}_2$ and determines whether $\mathcal{M}_1$ and $\mathcal{M}_2$ are equivalent.
\end{proposition}




The following proposition \cite{LQ08,QLMS15} is also useful in this paper.

\begin{proposition}\label{BLMTOBLM} Let BLM (RBLM) ${\cal M}$  have $n$ states and the alphabet $\Sigma\cup\{\tau\}$ where $\tau\notin\Sigma$. Then we can give another  BLM (RBLM) $\hat{\cal M}$ over the alphabet $\Sigma$ with the same states, such that $L_{{\cal M}}(w\tau)=L_{\hat{\cal M}}(w)$, for any $w\in \Sigma^*$.\end{proposition}

\begin{lemma}\cite{QLMS15}\label{1QFACTOBLM} For any given 1QFAC
\begin{equation}
{\cal M}=(S,Q,\Sigma, s_{0},|\psi_{0}\rangle, \delta,
\mathbb{U}, {\cal P}),
\end{equation}
 there is  a RBLM ${\cal M}'$ with
$(kn)^2$ states, where  $|S|=k$ and $|Q|=n$, such that

\begin{equation}
L_{{\cal M}}(x)=L_{{\cal M}'}(x)\label{Sim}
\end{equation}
for any $x\in\Sigma^*$.

\end{lemma}

By means of  Lemma \ref{1QFACTOBLM} and  Proposition \ref{BLMEQ}, we have the following theorem.
\begin{theorem}\label{1QFACEQ}
Two 1QFAC ${\cal M}_1$ and
${\cal M}_2$  are
equivalent if and only if they are $(k_1 n_1)^{2}+(k_2 n_2)^{2}-1$-equivalent.
Furthermore, there exists a polynomial-time algorithm running in
time $O([(k_1 n_1)^{2}+(k_2 n_2)^{2}]^3)$ that takes as input two 1QFAC ${\cal M}_1$
and ${\cal M}_2$ and determines whether ${\cal M}_1$ and ${\cal
M}_2$ are equivalent, where  $k_i$ and $n_i$ are the numbers of classical and quantum basis states of ${\cal M}_{i}$, respectively, $i=1,2$.

\end{theorem}


\section{Quantum Discrete Event Systems} \label{SECQDES}

In this section, we first recollect the  supervisory control of classical DES, 
then we  formalize QDES and corresponding  supervisory control  of QDES.   

\subsection{Language and Automaton Models of DES}\label{secDES}

In this subsection, we briefly review some basic concepts concerning DES \cite{book1,book2}. A DES is modeled and represented as a nondeterministic finite automaton $G$,
described by $G=(Q,\Sigma,\delta,q_{0},Q_{m})$,
where $Q$ is the finite set of states; $\Sigma$ is the finite set
of events; $\delta: Q\times \Sigma\rightarrow {\cal P}(Q)$ is the
transition function (In what follows, ${\cal P}(X)$
represents the power set of set $X$.); $q_{0}\in Q$ is the initial state; and
$Q_{m}\subseteq Q$ is called the set of marked states. Indeed,
transition function $\delta$ can be naturally extended to
$Q\times\Sigma^{*}$ in the following manner: For any $q\in Q$, any
$s\in\Sigma^{*}$ and $\sigma\in\Sigma$,
$
\delta(q,\epsilon)=\epsilon,\hskip 2mm {\rm and} \hskip 2mm
\delta(q,s\sigma)=\delta(\delta(q,s),\sigma),
$
where  we define $\delta(A,\sigma)=\bigcup_{q\in A}\delta(q,\sigma)$ for any $A\in {\cal P}(Q)$.

In particular, when $\delta$ is a map from $ Q\times \Sigma$ to $Q$, then it is a DFA, as we depict it in Fig. \ref{fig:DFA}.




In fact, in $G$ we can represent $q_{i}$ by vector $s_{i}=(0\hskip 2mm\cdots\hskip 2mm 1\hskip 2mm \cdots\hskip 2mm 0)^T$ where 1 is in the $i$th place and the dimension equals  $n$ ($T$ denotes transpose); for $\sigma\in\Sigma$, $\sigma$ is represented as a 0-1 matrix $(a_{ij})_{n\times n}$ where $a_{ij}\in \{0,1\}$, and $a_{ij}=1$ if and only if $q_{j}\in\delta(q_{i},\sigma)$. Analogously, vector $(0\hskip 2mm\cdots\hskip 2mm 1\hskip 2mm 0\hskip 2mm \cdots\hskip 2mm 1\hskip 2mm\cdots\hskip 2mm 0)^T$ in which 1 is in the $i$th and $j$th places, respectively, means that the current state may be $q_{i}$ or $q_{j}$.


For a DES modeled by finite automaton $G=(Q,\Sigma,\delta,q_{0},Q_{m})$,
$L(G)=\{x\in\Sigma^*: \delta(q_0,x)\in Q\}$ represents all feasible input strings in DES $G$,
and  $L_m(G)=\{x\in\Sigma^*: \delta(q_0,x)\in Q_m\}$ is called the language {\it marked} by $G$.



In order to define and better understand  parallel composition of
quantum  finite automata, we reformulate the parallel composition of
crisp finite automata \cite{book1,book2}. For finite automata
$G_{i}=(Q_{i},\Sigma_{i},\delta_{i},q_{0i},Q_{mi})$, $i=1,2$, we
reformulate the parallel composition in terms of the following
fashion:
\begin{align}
G_{1}\|G_{2}
=(Q_{1}\otimes Q_{2}, \Sigma_{1}\cup\Sigma_{2},
\delta_{1}\| \delta_{2},q_{10}\otimes q_{20},Q_{m1}\otimes
Q_{m2}).
\end{align}
Here, $Q_{1}\otimes Q_{2}=\{q_{1}\otimes q_{2}:q_{1}\in
Q_{1},q_{2}\in Q_{2}\}$, and symbol $``\otimes"$ denotes
tensor product. For event $\sigma\in
\Sigma_{1}\cup\Sigma_{2}$, we define the corresponding matrix of
$\sigma$ in $G_{1}\| G_{2}$ as follows:

(i) If event $\sigma\in \Sigma_{1}\cap \Sigma_{2}$, then $\sigma=\sigma_{1}\otimes\sigma_{2}$ where $\sigma_{1}$ and $\sigma_{2}$ are the
matrices of $\sigma$ in $G_{1}$ and $G_{2}$, respectively.

(ii) If event $\sigma\in\Sigma_{1}\backslash\Sigma_{2}$, then
$\sigma=\sigma_{1}\otimes I_{2}$ where $\sigma_{1}$ is the matrix
of $\sigma$ in $G_{1}$, and $I_{2}$ is unit matrix of order
$|Q_{2}|$.

(iii) If event $\sigma\in\Sigma_{2}\backslash\Sigma_{1}$, then $\sigma= I_{1}\otimes\sigma_{2}$ where $\sigma_{2}$ is the matrix of
$\sigma$ in $G_{2}$,  and $I_{1}$ is unit matrix of order $|Q_{1}|$.

In terms of the above (i-iii) regarding the event $\sigma\in  \Sigma_{1}\cup\Sigma_{2}$, we can define
$\delta_{1}\| \delta_{2}$ as:
For $ q_{1}\otimes q_{2} \in Q_{1}\otimes Q_{2}$, $\sigma\in \Sigma_{1}\cup\Sigma_{2}$,
\begin{align}
&(\delta_{1}\| \delta_{2})(q_{1}\otimes q_{2},\sigma)\\
=&\left\{
\begin{array}{lll}
(\sigma_{1}\otimes\sigma_{2})  \times(q_{1}\otimes q_{2} ),&{\rm if}& \sigma\in \Sigma_{1}\cap \Sigma_{2},\\
(\sigma_{1}\otimes I_{2})  \times (q_{1}\otimes q_{2} ),&{\rm if}& \sigma\in \Sigma_{1} \backslash \Sigma_{2},\\
(I_{1}\otimes\sigma_{2})\times (q_{1}\otimes q_{2} ),&{\rm if}& \sigma\in \Sigma_{2} \backslash \Sigma_{1},
\end{array}
\right.
\end{align}
where $\times$ is the usual product between matrices, and, as
indicated above, symbol $\otimes$ denotes tensor product
of matrices.


\subsection{Probabilistic DES (PDES)}

In the interest of  completeness, in this subsection, we would recall the automata model  for {\it probabilistic discrete event systems} (PDES). 
Formally, there are two definitions concerning PDES, and we name them  PDES-I and PDES-II respectively. 
    \begin{definition}\cite{LW1993,PPL2009,DYQ19}
        A PDES-I could be modeled as the following probabilistic automaton:
        \begin{equation}
             G = \{ X, x_{0}, \Sigma, \delta, \rho  \},
        \end{equation}
  where $X$ is the nonempty finite set of states; $x_{0} \in X $ is the initial state;
 $\Sigma$ is the nonempty finite set of events; $\Sigma = \Sigma_{c} \cup \Sigma_{uc}$ with  $\Sigma_{c}$ and $\Sigma_{uc}$ being the disjoint   controllable and uncontrollable event sets, respectively; 
        $\delta: X \times \Sigma \rightarrow X$ is the (partial) transition function, and the function $\delta$ can be extended to $X \times \Sigma^{*} $ by the natural manner; $\rho: X \times \Sigma \rightarrow [0,1]$ is the  transition-probability function, where
        $\rho(x,\sigma)$ is the probability of  transition $\delta(x,\sigma)$ and satisfies
         $\sum_{\sigma \in \Sigma} \rho(x,\sigma) \leq 1 $, $\forall x \in X$.
        In particular, if  $\sum_{\sigma \in \Sigma} \rho(x,\sigma) = 1$, $\forall x \in X$,
        then the system $G$ is called as a non-terminating PDES-I.

    \end{definition}

 PDES-II are defined in light of PFA in \cite{Paz71} that are from \cite{TT2005, KH2015, LQXF2008} as follows.
\begin{definition}
A PDES-II is a type of
systems with a quadruple
\begin{equation}
G=(Q,\Sigma,\eta, q_{0}),
\end{equation}
where $Q$ is a finite state space; $q_{0}\in Q $ is the initial
state; $\Sigma$ is a finite set of events; $\eta:
Q\times\Sigma\times Q\rightarrow [0,1] $ is a state transition
function: for $q, q'\in Q $ and $\sigma\in\Sigma$, $\eta(q,
\sigma, q')$ represents the probability that event $\sigma$
will occur, together with transferring the state of the machine
from $q$ to $q'$. If we require that $\sum_{q'\in Q}\eta(q,
\sigma, q')=1$ for any $q\in Q$ and any $\sigma\in\Sigma$, then it is exactly the PFA in \cite{Paz71} as defined in subsection \ref{DFAPFA}.
\end{definition}

However,  the problems of state complexity  in PDES-I and PDES-II  still need to be studied.






\subsection{Quantum  DES (QDES)}

As in \cite{MC00}, a {\it quantum language} over finite input alphabet $\Sigma$ is defined as a function mapping words to probabilities, i.e.,
a function from  $\Sigma^{\ast}$ to $[0,1]$. 

For any 1QFA ${\cal M}$ (MO-1QFA,1QFAC) with finite input alphabet $\Sigma$, the accepting probability $L_{\cal M}(x_1\dots x_n)$ for any $x_1\dots x_n\in \Sigma^*$ is defined as before. Therefore ${\cal M}$ {\it generates a quantum language} $L_{{\cal M}}$ over finite input alphabet $\Sigma$.

For any two quantum languages $f_1$ and $f_2$  over finite input alphabet $\Sigma$, denote $f_2\subseteq f_2$ if and only if  $f_2(w)\leq f_2(w)$ for any $w\in\Sigma^*$.


Denote
\begin{equation}
L_{{\cal M}}^{\lambda}=\{x\in\Sigma^{\ast}: f_{\cal M}(x)>\lambda\},
\end{equation}
where $0\leq\lambda<1$. Then
$L_{{\cal M}}^{\lambda}$ is called the language recognized by ${\cal M}$ \textit{with cut-point} $\lambda$.

A language, denoted by $L_{{\cal M}}^{\lambda,\rho}\subseteq \Sigma^{\ast}$, is recognized by ${\cal M}$ \textit{with some  cut-point $\lambda$ isolated by  some $\rho>0$}, if for any $x\in L_{{\cal M}}^{\lambda,\rho}$, $f_{{\cal M}} (x)\geq \lambda+\rho$ and for any $x\notin  L_{{\cal M}}^{\lambda,\rho}$, $f_{{\cal M}} (x)\leq \lambda-\rho$. In fact, if a language $L$ is recognized by a 1QFA ${\cal M}$ \textit{with some  cut-point $\lambda$ isolated by  some $\rho>0$}, then $L$ is also called to be recognized by a 1QFA ${\cal M}$ with {\it bounded error}.

In DES, the event set (input alphabet) $\Sigma$ is partitioned into two disjoint subsets $\Sigma_c$ (controllable  events) and  $\Sigma_{uc}$ (uncontrollable events), and a specification language $K\subset \Sigma^{\ast}$ is given. It is assumed that  controllable  events can be disabled by a supervisor. To solve the supervisory control problem we need to find a supervisor for performing a feedback control with the plant that is described by an automaton.
 Formally, supervisor $S$ is defined as a function:
\[
S: L(G)\rightarrow {\cal P}(\Sigma).
\]
It is interpreted that for each $s\in L(G)$, $S(s)\cap
\{\sigma:s\sigma\in L(G)\}$ represents the set of enabled events
after the occurrence of $s$. Furthermore, it is required that for
any $s\in L(G)$,
\begin{equation}\label{CS}
\Sigma_{uc}\cap \{\sigma\in \Sigma:s\sigma\in L(G)\}\subseteq S(s),
\end{equation}
which means that after the occurrence of any physically possible
string of events, the physically possible  uncontrollable events
are not allowed to be disabled by $S$. The condition described by
Eq. (\ref{CS}) is called {\it admissible} for $S$ \cite{book1,book2}.

A QDES is a quantum system (called a quantum plant) described by a 1QFA  ${\cal M}$ together with the event set $\Sigma=\Sigma_c\cup\Sigma_{uc}$
and simulated as the quantum language generated by this 1QFA ${\cal M}$ (sometimes  simulated as the quantum language recognized by this 1QFA ${\cal M}$ with some cut-point $\lambda$ or with some  cut-point $\lambda$ isolated by  some $\rho>0$).

A quantum supervisor ${\cal S}$ for controlling ${\cal M}$ is defined formally as a function ${\cal S}:  \Sigma^\ast \rightarrow [0,1]^\Sigma$, where for any $s\in\Sigma^{\ast}$, ${\cal S}(s)$ is a quantum language over $\Sigma$. Intuitively, after inputting $s$ in QDES ${\cal M}$, for any  $\sigma\in\Sigma$, ${\cal S}(s)(\sigma)$ denotes the degree to which $\sigma$ is enabled.

Therefore, we require that
 quantum supervisor ${\cal S}$ logically satisfies: $\forall \sigma\in\Sigma_{uc}$, $\forall s\in \Sigma^{\ast}$,
\begin{equation} \label{admissible}
L_{{\cal M}}(s\sigma)\leq {\cal S}(s)(\sigma),
\end{equation}
 which denotes that both $s$ and $s\sigma$ being feasible in the quantum plant for uncontrollable event $\sigma$ results in $\sigma$ being enabled after quantum supervisor ${\cal S}$ controlling $s$. Equation (\ref{admissible}) is called \textit{quantum admissible condition}.

 The feedback loop of supervisory control of QDES ${\cal M}$ controlled by quantum supervisor ${\cal S}$ can be depicted as Fig. \ref{fig:QDES}.

In classical DES, given a DES modeled by finite automaton $G$ and admissible
supervisor $S$, the resulting controlled system denoted by $S/G$
that means $S$ controlling $G$, is also a DES modeled by a
language $L(S/G)$ defined recursively as follows \cite{book1, book2}: $\epsilon \in
L(S/G)$ and
\vskip 2mm
$s\sigma\in L(S/G)\Longleftrightarrow s\in L(S/G)$ and $s\sigma \in  L(G)$ and $\sigma\in S(s)$.\\

Naturally, we denote by ${\cal S}/{\cal M}$ the controlled system by ${\cal S}$, 
and $L_{{\cal S}/{\cal M}}$  is defined as a function from $\Sigma^{\ast}$ to $[0,1]$  as follows:



First, it is required that $L_{{\cal S}/{\cal M}} (\epsilon)=1$ (i.e., the starting state is an accepting state) 
 and then recursively,  $\forall s\in\Sigma^{\ast}$, $\forall \sigma\in\Sigma$, the following equation holds
 \begin{equation}
 L_{{\cal S}/{\cal M}}(s\sigma)=\min\{L_{{\cal S}/{\cal M}}(s), L_{{\cal M}}(s\sigma),{\cal S}(s)(\sigma)\}.
 \end{equation}
 By intuitive the above equation logically implies that $s\sigma$ can be performed by the controlled system ${\cal S}/{\cal M}$ if and only if $s$ can be performed by the controlled system ${\cal S}/{\cal M}$ and $s\sigma$ is feasible in the quantum plant as well as $\sigma$ is enabled by the quantum supervisor ${\cal S}$ after the event string $s$ occurs.



In classical DES and automata theory, the prefix-closure $pr(L)$ of a language $L$ over alphabet $\Sigma$ is defined as: for any $s\in\Sigma^{\ast}$, $s\in pr(L)$ if and only if $st\in L$ for some $t\in\Sigma^{\ast}$. Naturaly,
if ${\cal K}$ is a quantum language over alphabet $\Sigma$, then logically we define the quantum language of its prefix-closure $pr({\cal K})$ as follows: $\forall s\in \Sigma^{\ast}$,
\begin{align}
     pr({\cal K})(s)=\sup_{t\in\Sigma^{\ast}}{\cal K}(st).
 \end{align}


\begin{figure}[htbp]
\centering
\includegraphics[width=0.55\textwidth]{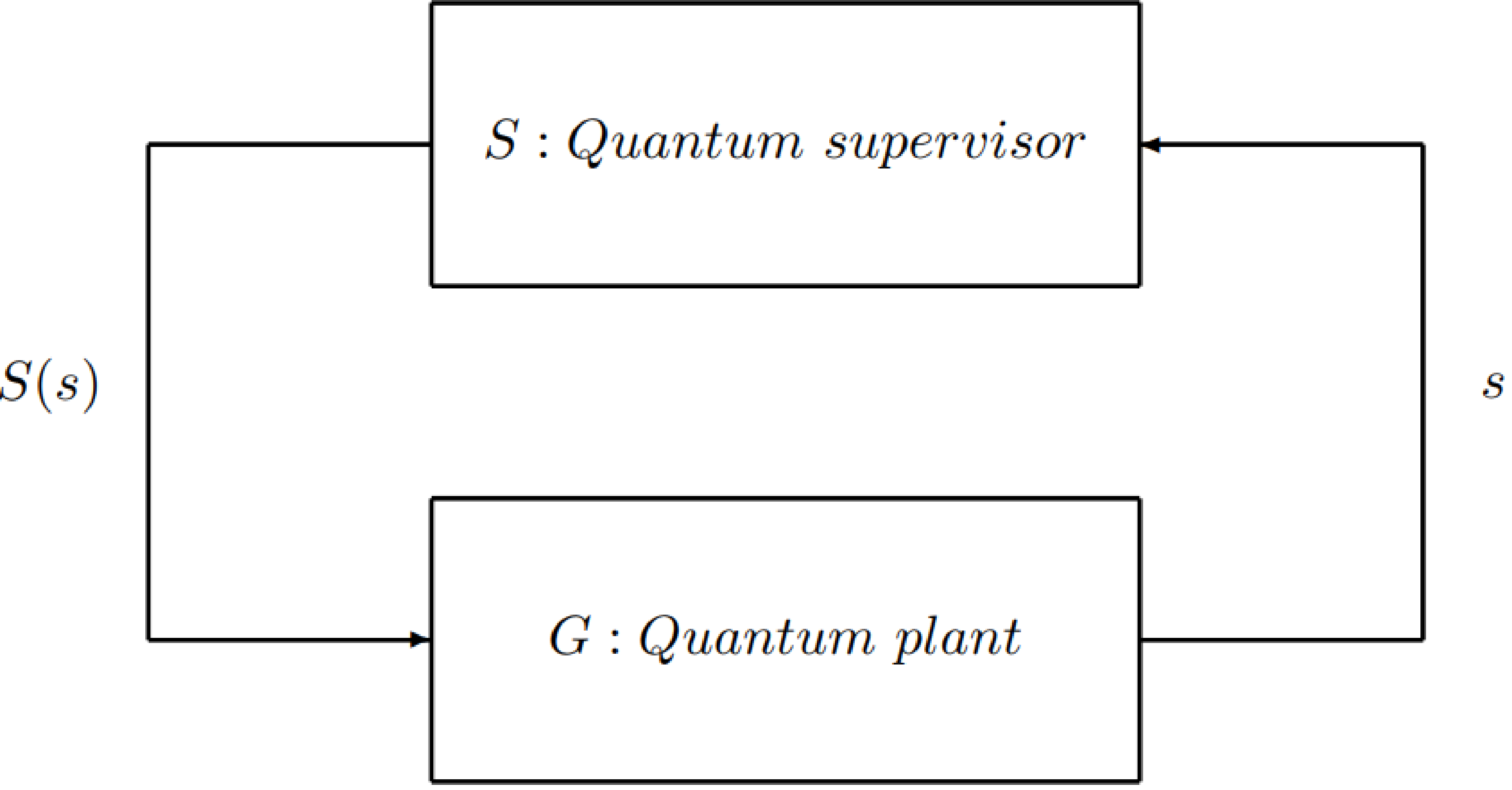}
\caption{The supervisory control of QDES, where $G$ represents the uncontrolled system and $S$ represents the quantum supervisor.}
\label{fig:QDES}
\end{figure}



\subsection{Parallel composition of QDES}

Let two QDES with the same finite event set (input alphabet) $\Sigma$ be described by two 1QFAC ${{\cal M}_i}=(S_i,Q_i,\Sigma, s_{0}^{(i)},|\psi_{0}^{(i)}\rangle, \delta_i,
\mathbb{U}_i, {\cal P}_i)$, $i=1,2.$ Then the parallel composition of QDES ${{\cal M}_1}$ and ${{\cal M}_2}$ is their tensor operation, that is the 1QFAC ${{\cal M}_1}\otimes {{\cal M}_2}$ as follows.
\begin{align}
&{{\cal M}_1}\otimes {{\cal M}_2}=(S_1\otimes S_2, Q_1\otimes Q_2, \Sigma, (s_{0}^{(1)}, s_{0}^{(2)}), \nonumber\\
&|\psi_{0}^{(1)}\rangle\otimes |\psi_{0}^{(2)}\rangle,
\delta_1 \otimes \delta_2,
 \mathbb{U}_1\otimes \mathbb{U}_2, {\cal P}_1\otimes {\cal P}_2),
\end{align}
 where
 \begin{itemize}
 \item  $S_1\otimes S_2=\{(s_1,s_2): s_i\in S_i, i=1,2\}$;

  \item $Q_1\otimes Q_2$ means the set $\{|q_{1,i}\rangle\otimes |q_{2,j}\rangle: q_{1,i}\in Q_1, q_{2,j}\in Q_2\}$;

  \item $\mathbb{U}_1\otimes \mathbb{U}_2=\{U_{s_i\sigma}\otimes U_{s_j\sigma}: (s_i,s_j)\in S_1\otimes S_2, U_{s_i\sigma}\in \mathbb{U}_1, U_{s_j\sigma}\in \mathbb{U}_2, \sigma\in\Sigma \}$;

   \item ${\cal P}_1\otimes {\cal P}_2=\{{\cal P}_{(s_i,s_j)}:(s_i,s_j)\in S_1\otimes S_2\}$ and ${\cal P}_{(s_i,s_j)}=\{P_{(s_i,s_j),a},P_{(s_i,s_j),r}\}$;

   \item $\delta_1 \otimes \delta_2 ((s_1,s_2),\sigma)=(\delta_1(s_1,\sigma), \delta_2(s_2,\sigma))$.
\end{itemize}
 It is easy to check that for any $s=x_1x_2\ldots x_n\in \Sigma^{\ast}$,
\begin{align}
&L_{  {{\cal M}_1}\otimes {{\cal M}_2}}(s)\nonumber\\
=& \|P_{(\mu_1(s),\mu_2(s)),a}v_1(s)\otimes v_2(s)|\psi_{0}^{(1)}\rangle\otimes |\psi_{0}^{(2)}\rangle\|^{2}\\
=&\|P_{\mu_1(s),a}v_1(s)|\psi_{0}^{(1)}\rangle\|^{2}\|P_{\mu_2(s),a}v_2(s)|\psi_{0}^{(2)}\rangle\|^{2}\\
=&L_{{\cal M}_1}(s) L_{{\cal M}_2}(s),
\end{align}
where
\begin{align}
v_i(s)=&U_{\mu_i( x_{1}\cdots x_{n-2}x_{n-1} )x_{n}}U_{\mu_i(x_{1}\cdots
 x_{n-3}x_{n-2})x_{n-1}}\cdots   \nonumber \\
 & U_{\mu_i(x_{1})x_{2}}
U_{s_{0}^{(i)}x_{1}},
\end{align}
and $\mu_i(s)=\delta_i^*(s_{0}^{(i)},s)$, $i=1,2.$

\section{Supervisory Control of QDES} \label{SECQSCT}

In this section, we first present some properties and new examples concerning 1QFA, and these results are new and useful for the study of supervisory control of QDES. In DES, the occurrence of an event string $s=x_1x_2\ldots x_n$ being feasible entails usually that any prefix of $s$ is feasible as well \cite{book1,book2}. So, 1QFA used for simulating QDES need to recognize with cut-point or bounded error prefix-closured languages. However, we will prove that any MO-1QFA is short of this ability.   MM-1QFA can recognize some prefix-closured languages, but cannot recognize all regular languages with bounded error.  Therefore, 1QFAC  are better for simulating QDES. QDES simulated by 1QFAC can be thought of as  hybrid systems of quantum and classical control since 1QFAC are an integration of MO-1QFA and DFA.

\subsection{Some properties and new examples concerning QFA}

First we present a result from \cite{BP02}.

{\bf Fact 1}.  \cite{BP02} For any unitary matrix $U$ and any $\epsilon>0$ there exists
an integer $n>0$ such that $\| I-U^n\|_2< \epsilon$.

We call a language $L$ is prefix-closured if for any $s\in L$,  any prefix of $s$ also belongs to $L$.
From the above fact we can prove that no MO-1QFA can recognize prefix-closured languages. That is the following Fact.

{\bf Fact 2}.  Let $\Sigma$ be a finite alphabet, and let $L\subsetneq\Sigma^*$ be any regular language with prefix closure. Then no MO-1QFA can recognize $L$ with cut-point or bounded error.

{\bf Proof:}
First we note that empty string $\epsilon\in L$. If any $s\in L$ and any $\sigma\in \Sigma$ imply $s\sigma\in L$, then it is easy to see $L=\Sigma^*$. So, there exist $s\in L$ and $\sigma\in \Sigma$ such that $s\sigma\notin L$, and therefore $s\sigma^k\notin L$ for any $k\geq 1$. If there exist an MO-1QFA ${\cal M}=(Q, |\psi_0\rangle,
\{U(\sigma)\}_{\sigma\in \Sigma}, Q_a,Q_r)$ and a cut-point $0\leq\lambda<1$ such that  ${\cal M}$ recognizes $L$ with cut-point $\lambda$, then $\|P(a)U_s|\psi_0\rangle\|^2>\lambda$ due to $s\in L$. By virtue of {\bf Fact 1}, there is $k\geq 1$ such that $\|U_{\sigma}^k-I\|_2<\|P(a)U_s|q_1\rangle\|-\sqrt{\lambda}$, and therefore we have
\begin{align}
\|P(a)U_{s{\sigma}^k}|q_1\rangle-P(a)U_s|q_1\rangle\|\leq &\|U_{s{\sigma}^k}|q_1\rangle-U_s|q_1\rangle\|\\
\leq&\|U_{\sigma}^k-I\|_2\\
<&\|P(a)U_s|q_1\rangle\|-\sqrt{\lambda},
\end{align}
 which results in $\|P(a)U_{s{\sigma}^k}|q_1\rangle\|>\sqrt{\lambda}$, implying $s{\sigma}^k\in L$, a contradiction. So, we  have no MO-1QFA recognizing $L$ with cut-point (or bounded error).
\hskip 7mm $\blacksquare$

Now we present the first prefix-closured regular language that shows the state complexity advantage of 1QFAC over PFA. 

\begin{example}\label{EGnew} For any $m\in \mathbb{Z}^+$, let
\begin{equation} \label{Lm}
L(m)=\{w\in\{0,1\}^*:0^{km}1 \text{ is not a prefix of } w,\forall k\in \mathbb{N} \}
\end{equation}
where $\mathbb{N}$ is the set of all natural numbers, and $ \mathbb{Z}^+$ is the set of all positive integers. 
\end{example}

We have the following theorem.

\begin{theorem} \label{EGtheorem}
There exists an 1QFAC recognizing language $L(m)$ (Eq. (\ref{Lm})) with cut-point $0$ and with $2$ classical states and $2$ quantum states, but any PFA recognizing  $L(m)$ with cut-point $0$ has at least $\lceil\log_2m\rceil$ states.
\end{theorem}

Next we prove the theorem in detail. First we need a lemma.

\begin{lemma}\label{lower_bound}
Let $L$ be a language over $\Sigma=\{0\}$. Suppose the minimal DFA recognizing $L$ has $m$ states. Then any PFA recognizing $L$  with cut-point $0$ has at least $\lceil\log_2m\rceil$ states.
\end{lemma}
\begin{proof}
   Let $\mathcal{M}=(S,\Sigma,\pi,\{M(\sigma)\}_{\sigma\in\Sigma},\eta)$ be a PFA that recognizes $L$ with cut-point $0$, where $S=\{s_1,s_2,\cdots,s_{|S|}\}$. Denote by $p_{ij}$ the $(i,j)$-th entry of $M(0)$. For any $s_i\in S$, define $\delta$ a transfer relation of states  as follows:
$$\delta(\{s_i\},0)=\{s_j:p_{ij}>0\};$$ 
for any $S_1\subseteq S$,  define 
$$\delta(S_1,0)=\mathop{\bigcup}\limits_{s\in S_1}\delta(\{s\},0);$$ 
for any $S_1\subseteq S$, $x\in \Sigma^*$, recursively  define 
$$\delta(S_1,0x)=\delta(\delta(S_1,0),x).$$ 
Let $S_0=\{s\in S:\text{the initial probability of $s$ is greater than $0$}\}$ and let $F$ be the set of accepting states of $\mathcal{M}$. It is easy to check that $L_{\mathcal{M}}^{0}=\{x\in\Sigma^*:\delta(S_0,x)\cap F\neq\varnothing\}$.

   Based on the above facts, we can construct a DFA  recognizing $L$ with $2^{|S|}$ states. Let $\mathcal{A}=(\mathcal{P}(S),\Sigma,\delta,S_0,Q_a)$ be a DFA, where $Q_a=\{S_1\subseteq S:S_1\cap F\neq\varnothing\}$. We can see that $x\in\Sigma^*$ is accepted by $\mathcal{A}$, iff $\delta(S_0,x)\in Q_a$,  iff $\delta(S_0,x)\cap F\neq\varnothing$, and iff $x\in L_{\mathcal{M}}^{0}$. Thus $\mathcal{A}$ recognizes $L_{\mathcal{M}}^{0}$.
   
   Since the minimal DFA of $L$ has $m$ states and $\mathcal{A}$ has $2^{|S|}$ states, we get that $2^{|S|}\geq m$. Therefore $|S|\geq \lceil\log_2m\rceil$.  
\quad  $\blacksquare$
\end{proof}

From the above theorem we have  a corollary.

\begin{corollary}\label{corollary1}
Let $L=\{0^{t}:t \bmod m\neq 0\}$ be a language over $\Sigma=\{0\}$. Then any PFA recognizing $L$  with cut-point $0$ has at least $\lceil\log_2m\rceil$ states, where $m\in\mathbb{Z}^+$.
\end{corollary}
\begin{proof}
Immediate from Lemma \ref{lower_bound}.  
\quad $\blacksquare$
\end{proof}

Due to Corollary \ref{corollary1} we can obtain the following proposition.


\begin{proposition}\label{p1}
Any PFA recognizing language $L(m)$ with
 cut-point $0$ has at least $\lceil\log_2m\rceil$ states, where $m\in\mathbb{Z}^+$.
\end{proposition}
\begin{proof} We prove by contradiction. 
Let $\mathcal{M}=(S,\Sigma,\pi,\{M(\sigma)\}_{\sigma\in\Sigma},\eta)$ be a PFA that recognizes $L(m)$ with cut-point $0$ and with less than $\lceil\log_2m\rceil$ states, where $\Sigma=\{0,1\}$. Define $\chi:\mathbb{R}^{|S|}\rightarrow \mathbb{R}^{|S|}$ as: for any $\eta_1\in {\mathbb{R}}^{|S|}$, denote by $a_i,b_i$  the $i$-th entries of $\eta_1$ and $\chi(\eta_1)$, respectively, where 
\begin{equation}
 b_i=
\begin{cases}
        0, &\text{if $a_i=0$},\\
        1,&\text{if $a_i\neq 0$}.
\end{cases}
\end{equation}
Since $\mathcal{M}$ recognizes $L(m)$ with cut-point $0$,  we have $\pi M(0)^{t}M(1)\eta>0$ iff $t \bmod m\neq 0$. Since $\pi M(0)^{t}M(1)\eta>0$ iff $\pi M(0)^{t}\chi(M(1)\eta)>0$, we have $\pi M(0)^{t}\chi(M(1)\eta)>0$ iff $t \bmod m\neq 0$. Thus, PFA $\mathcal{M}'=(S,\{0\},\pi,\{M(\sigma)\}_{\sigma\in\{0\}},\chi(M(1)\eta))$ can recognize unary language $L=\{0^{t}:t \bmod m\neq 0\}$ with cut-point $0$. It contradicts Corollary \ref{corollary1}. Therefore, the proposition holds. \quad $\blacksquare$\end{proof}

Now we prove the state complexity advantage of  1QFAC over PFA for recognizing $L(m)$.
\begin{proposition}\label{p2}
There exists an 1QFAC recognizing language $L(m)$ with cut-point $0$ and with $2$ classical states and $2$ quantum states, where $m\in\mathbb{Z}^+$.
\end{proposition}
\begin{proof}
The idea for constructing a 1QFAC  is as follows. For any input string $w\in\Sigma^*$, we use classical states to determine whether $w$ starts with $0^n1$, where $n\in\mathbb{N}$. If not, we accept it. Otherwise we measure quantum states to determine whether $n \bmod m=0$. So,  1QFAC $\mathcal{M}=(S,Q,\Sigma,s_0,|q_0\rangle,\delta,\mathbb{U},\mathcal{P})$ is constructed as follows, where 
\begin{itemize}
\item $S=\{s_0,s_{f}\}$;
\item $Q=\{q_0,q_1\}$;
\item $\Sigma=\{0,1\}$;
\item $\delta(s_0,0)=s_0,\delta(s_0,1)=s_{f}, \delta(s_{f},0)=\delta(s_f,1)=s_f$;
\item $U_{s_0,0}=\begin{bmatrix}
                        \cos{\frac{\pi}{m}}& \sin{\frac{\pi}{m}}  \\
                        -\sin{\frac{\pi}{m}} & \cos{\frac{\pi}{m}}
                   \end{bmatrix}$ and other operators in $\mathbb{U}$ are $I$;
\item $P_{s_0,acc}=I,P_{s_0,rej}=O$; $P_{s_f,acc}=|q_1\rangle\langle q_1|,P_{s_f,rej}=|q_0\rangle\langle q_0|$.
\end{itemize}
It works as follow. For any input string $w\in\Sigma^*$, if $w$ does not start with $0^n1$, where $n\in\mathbb{N}$, that is $w\in\{0\}^*$, then the final classical state of $\mathcal{M}$ is $s_0$ and $w$ will be accepted with probability $1$. If $w$ starts with $0^n1$, then the final classical state of $\mathcal{M}$ is $s_f$. The final quantum state is $U_{s_0,0}^n|q_0\rangle$. Note that
\begin{equation}
U_{s_0,0}^n=\begin{bmatrix}
                        \cos{\frac{\pi n}{m}}& \sin{\frac{\pi n}{m}}  \\
                        -\sin{\frac{\pi n}{m}} & \cos{\frac{\pi n}{m}}
                   \end{bmatrix}.\
\end{equation}
We obtain that the accepting probability $L_{\cal M}(w)$ is 
\begin{align}
  L_{\cal M}(w)  &=\|P_{s_f,acc}U_{s_0,0}^n|q_0\rangle\|^2\\
&=\sin^2{\frac{\pi n}{m}}.
\end{align}
If $n \bmod m=0$, then $L_{\cal M}(w)=0$, otherwise $L_{\cal M}(w)>0$. Therefore, the proposition holds. 
\quad $\blacksquare$

\end{proof}

By combining Propositions  \ref{p1} and \ref{p2},  we have proved Theorem \ref{EGtheorem}.

\begin{remark}
For any $p,q\in \mathbb{Z}^+$ and $p,q\geq 2$, it is clear that $L(p)\subset L(pq)$ holds.

\end{remark}

\vskip 5mm

In addition, we give another example   $L^{(N)}$ that shows 1QFAC have the state complexity advantage over DFA. We can construct  a PFA that requires more states than 1QFAC to recognize $L^{(N)}$,
 but we still do not know the minimal state number of PFA  recognizing $L^{(N)}$.

\begin{example} \label{EG1} Let $\Sigma=\{0,1,2\}$. Given a natural number $N$,
\begin{align}
L^{(N)}=&\{w\in \Sigma^{\ast}: |w_{0,1}|< 2N\}
\cup\{w\in \Sigma^{\ast}:    \nonumber  \\
& |w_{0,1}|=2N,
  w_{0,1}=x_1x_2\cdots x_Ny_{1}y_2\cdots y_{N}\}, 
\end{align}
 where \begin{equation}\sum_{i=1}^{N}x_i2^{N-i}+\sum_{i=1}^{N}y_i2^{N-i}=2^N-1 \end{equation} and
$w_{0,1}$ denotes the substring of $w$ by removing all $2$ in $w$.

\end{example}

\begin{remark}
By means of Myhill-Nerode theorem \cite{HU79},  we can know that DFA require  $\Omega(2^N)$ states to recognize the language $L^{(N)}$. In fact, as usual,  define the equivalence relation $\equiv_{L^{(N)}}$ over $\Sigma^*$: for any $x,y\in \Sigma^*$, $x\equiv_{L^{(N)}}y$ if and only if for any $z\in \Sigma^*$, $xz\in L^{(N)} \Leftrightarrow yz\in L^{(N)}$. For any $x,y\in \{0,1\}^*$, with $|x|=|y|=N$ and $x\neq y$, then there is $z\in \{0,1\}^*$ with $|z|=N$ such that $x+z=2^N-1$ (i.e., $xz\in L^{(N)}$). However, $yz\notin L^{(N)}$ since $x+z\neq y+z$ due to $x\neq y$. So, $x \notequiv_{L^{(N)}}y$ and the number of equivalence classes is at least $|\{0,1\}^N|=2^N$. As a result, the number of states of any DFA recognizing $L^{(N)}$ is at least $2^N$ as well.

\end{remark}

For 1QFAC to recognize $L^{(N)}$, we have the following result.

\begin{theorem}
 For any $0<\epsilon<1$, there exists a 1QFAC ${\cal M}$ having  $2N+2$ classical states and $\Theta(N)$ quantum basis states to  recognize $L^{(N)}$, satisfying $L_{{\cal M}}(w)=1$ for every $w\in L^{(N)}$, and $L_{{\cal M}}(w)<\epsilon$ for every $w\in \Sigma^*\setminus L^{(N)}$. 
\end{theorem}


\begin{proof} ${\cal M}$ can be constructed as follows. First we need to employ an important result by Ambainis and Freivalds \cite{AF98}: For language $\{0^{kp}: k\in \mathbb{N}\}$, where $p$ is a prime number and $2^{N+1}<p<2^{N+2}$, then there is an MO-QFA ${\cal M}_0$ recognizing $\{0^{kp}: k\in \mathbb{N}\}$, say ${\cal M}_0=(Q, |\psi_0\rangle,
\{U(0)\}, Q_a,Q_r)$, where $U(0)^p=I$ and $|Q|=\Theta(\log p)$. Then
$\|P(a)U(0)^t|\psi_0\rangle\|^2=1$ for $t=kp$ with some $k\in\mathbb{N}$, and $\|P(a)U(0)^t|\psi_0\rangle\|^2<\epsilon$ for $t\neq kp$ with any $k\in\mathbb{N}$.

1QFAC ${\cal M}=(S,Q,\Sigma, s_{0},|\varphi_{0}\rangle, \delta,
\mathbb{U}, {\cal P})$ can be constructed as:
\begin{itemize}
\item $S=\{s_i:i=0,1,\ldots,2N+1\}$;
\item For $\sigma\in\{0,1\}$, $\delta(s_i,\sigma)=s_{i+1}$ for $i\leq 2N$; $\delta(s_{2N+1},\sigma)=s_{2N+1}$, and $\delta(s_i,2)=s_i$ for any $s_i\in S$;

\item ${\cal P}=\{{\cal P}_{s_i}: s_i\in S\}$ where ${\cal P}_{s_i}=\{P_{s_i,a},P_{s_i,r}\}$ and $P_{s_i,a}=I$ for $i<2N$, $P_{s_{2N},a}=P(a)$, $P_{s_{2N+1},r}=I$;

\item $|\varphi_{0}\rangle=U(0)^{p-2^{N}+1}|\psi_0\rangle$;
\item $\mathbb{U}=\{U_{s\sigma}:s\in S,\sigma\in\Sigma\}$ where  $U_{s\sigma}=U(0)^{\sigma 2^{N-1-i}}$ for $\sigma\in\{0,1\}$, $s=s_i$ or $s=s_{N+i}$,  with $i\leq N-1$; $U_{s_{2N}\sigma}$ and $U_{s_{2N+1}\sigma}$ can be any unitary operator for $\sigma\in\{0,1\}$, $U_{s2}=I$ for any $s\in S$.
\end{itemize}

In the light of the above constructions, for $w\in\Sigma^*$, if $|w_{0,1}|< 2N$ then $w$ is accepted exactly;  if  $|w_{0,1}|> 2N$ then $w$ is rejected exactly; if $|w_{0,1}|= 2N$, denote $w_{0,1}= x_1x_2\cdots x_Ny_{1}y_2\cdots y_{N}$, then the classical state is $s_{2N}$, and the quantum state is
\begin{align}
&U_{s_{2N-1}y_N}\cdots U_{s_Ny_1}U_{s_{N-1}x_N}\cdots U_{s_1x_2}U_{s_0x_1}|\varphi_{0}\rangle\\
=&U(0)^{p-2^N+1+\sum_{i=1}^{N}x_i2^{N-i}+\sum_{i=1}^{N}y_i2^{N-i}}|\psi_0\rangle,\\
=&|\psi(w)\rangle,
\end{align}
and the accepting probability is
\begin{equation}
\|P(a)|\psi(w)\rangle\|^2.
 \end{equation}
 So, $L_{{\cal M}}(w)=1$ for $w\in L^{(N)}$ and $L_{{\cal M}}(w)\leq \epsilon$ for $w\notin L^{(N)}$.

 In fact, after reading input symbol 2, neither the classical nor quantum states have been changed, so without loss of generalization, we consider the dynamics of ${\cal M}$ for computing string $\sigma_0\sigma_1\ldots\sigma_{2N-1}\in\{0,1\}^*$, and it is depicted by Fig. \ref{fig:example1}.
\quad $\blacksquare$

\begin{figure}[htbp]
\centering
\includegraphics[width=0.9\textwidth]{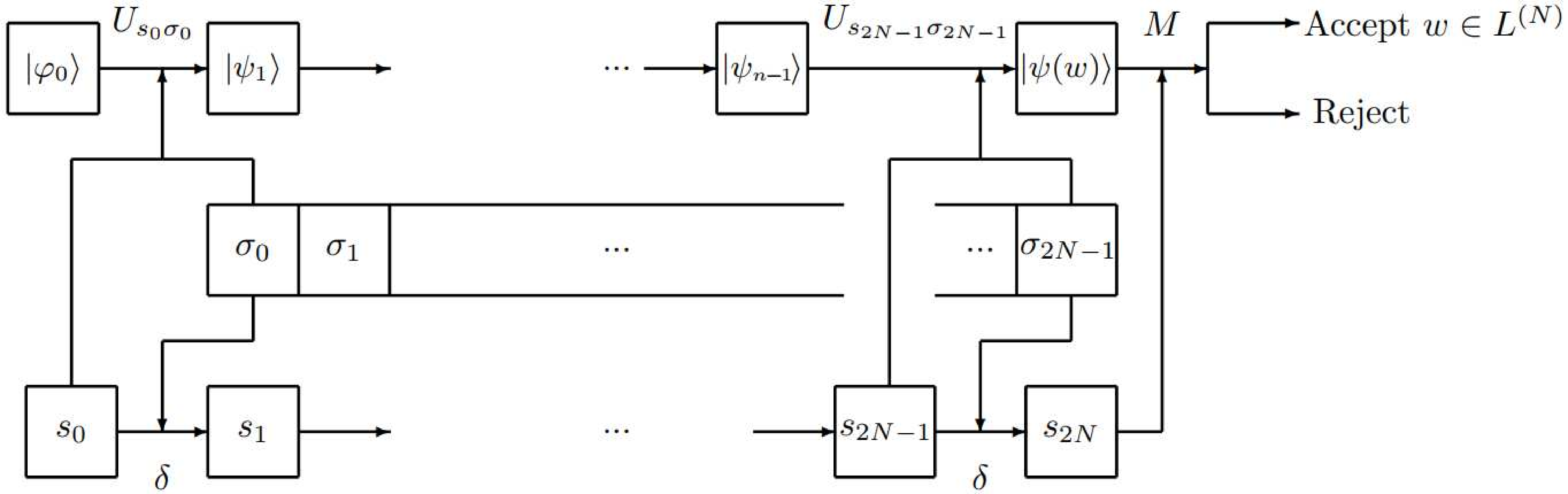}
\caption{1QFAC dynamics ${\cal M}$ for input string $\sigma_0\sigma_1\ldots\sigma_{N-1}\in\{0,1\}^*$.}
\label{fig:example1}
\end{figure}

\end{proof}

For PFA to recognize $L^{(N)}$, we have:

\begin{theorem}
 For any $0<\epsilon<1$, 
there exists a PFA   ${\cal M'}$ with $O(N^3)$ states recognizing 
$L^{(N)}$,  satisfying that the accepting probability $P_{{\cal M'}}(w)> 1-\epsilon$ for  $w\in  L^{(N)} $ and $P_{{\cal M'}}(w)<\epsilon$ for  $w\notin  L^{(N)} $.

\end{theorem}

\begin{proof}
Due to Theorem 10 by Ambainis and Freivalds in \cite{AF98}, there exists a PFA $\mathcal{M}=(S,\{0\},\pi,\{M(\sigma)\}_{\sigma\in\{0\}},\eta)$ with $O(N^2)$ states recognizing $L=\{0^{2^{N}-1}\}$  with probability $1-\epsilon$, that is,  $\pi M(0)^{2^N-1}\eta>1-\epsilon$ and $\pi M(0)^{t}\eta<\epsilon$ for $t\neq 2^N-1$. We construct a PFA $\mathcal{M}'=(S',\{0,1,2\},\pi',\{M'(\sigma)\}_{\sigma\in\{0,1,2\}},\eta')$ recognizing $L^{(N)}$  with probability $1-\epsilon$, where $M'(2)$ is an identity operator and the others are defined as follows:
\begin{itemize}
\item $S'=\{(s,q):s\in\{s_0,\cdots,s_{2N+1}\}, q\in S \}$, and for convenience, we denote the states in $S'$ as $\langle s|\langle q|$;
\item $\pi'=\langle s_0|\otimes \pi$;

\item $$\langle s_i|\langle q|M'(\sigma)=\langle s_{i+1}|\left(\langle q|M(0)^{\sigma 2^{N-1-i}}\right),$$
$$\langle s_{N+i}|\langle q|M'(\sigma)=\langle s_{N+i+1}|\left(\langle q|M(0)^{\sigma 2^{N-1-i}}\right),$$
$$\langle s_{2N}|\langle q|M'(\sigma)=\langle s_{2N+1}|\langle q|M'(\sigma)=\langle s_{2N+1}|\langle q|,$$
$i=0,\cdots, N-1$, $\sigma\in\{0,1\}$, $q\in S$;
\item $\eta'=(\sum\limits_{i=0}^{2N-1}|s_i\rangle)\otimes |e\rangle+|s_{2N}\rangle\otimes\eta$, where $|e\rangle$ represents a vector in $\mathbb{R}^{|S|}$ that all of its elements are 1.
\end{itemize}

For any input string $w\in\Sigma^*$, $\mathcal{M}'$ works as follows. According to the definition of $\eta'$, if $|w_{0,1}|<2N$, then the resulting state of $\mathcal{M}'$ computing $w$ is $\langle s_i|\langle q|$,  where $i<2N$ and $\langle q|$ is a superposition state. Hence $w$ will be accepted with probability 
$$\langle s_i|\langle q|\eta'=1.$$
If $|w_{0,1}|>2N$, then the state of $\mathcal{M}'$ is $\langle s_{2N+1}|\langle q|$, where $\langle q|$ is a superposition state. Hence the accepting probability of $w$ is 
$$\langle s_{2N+1}|\langle q|\eta'=0.$$
If $|w_{0,1}|=2N$, that is $w=x_1\cdots x_N y_1\cdots y_N$, then  we have 
\begin{align}
&\pi'M'(x_1)\cdots M'(x_N) M'(y_1)\cdots M'(y_N)\eta'\\
=&(\langle s_0|\otimes \pi) M'(x_1)\cdots M'(x_N) M'(y_1)\cdots M'(y_N)\eta'\\
=&\left( \langle s_{2N}| \pi M(0)^{\sum\limits_{i=1}^N x_i2^{N-i}+\sum\limits_{i=1}^N y_i2^{N-i}}\right)\eta' \\
=&\left(\langle s_{2N}| \pi M(0)^{\sum\limits_{i=1}^N x_i2^{N-i}+\sum\limits_{i=1}^N y_i2^{N-i}}\right)|s_{2N}\rangle\otimes \eta \\
=& \pi M(0)^{\sum\limits_{i=1}^N x_i2^{N-i}+\sum\limits_{i=1}^N y_i2^{N-i}}\eta .
\end{align} 
Hence, if $\sum\limits_{i=1}^N x_i2^{N-i}+\sum\limits_{i=1}^N y_i2^{N-i}=2^N-1$, $w$ will be accepted with probability greater than $1-\epsilon$, otherwise $w$ will be accepted with probability less than $\epsilon$.
\quad $\blacksquare$

\end{proof}






\begin{remark}

 So, for prefix-closured regular language $L^{(N)}$,  there exists a 1QFAC ${\cal M}$ having  $2N+2$ classical states and $\Theta(N)$ quantum basis states to  recognize $L^{(N)}$ with bounded error, and 
there exists a PFA   ${\cal M'}$ with $O(N^3)$ states recognizing 
$L^{(N)}$ with bounded error.  However, we still do not know the lower bound on the state complexity of PFA for recognizing  $L^{(N)}$. 
\quad $\blacksquare$

\end{remark}

\subsection{Supervisory Control of QDES simulated with cut-point languages}





 Let ${\cal M}$ be a 1QFAC, and
 let $K\subseteq L_{{\cal M}}^{\lambda}$ ($0\leq\lambda<1$) be the set of specifications we hope to achieve, where $K$ is also a regular language that is recognized by a 1QFAC ${\cal H}$ with  cut-point $\mu$ ($\mu\geq \lambda$) isolated by $\rho$. First, we give a sufficient condition such that there is a quantum supervisor controlling QDES ${\cal M}$ to approximate to $K$, and this is the first supervisory control theorem of QDES.

From now on,  a QDES associated with a 1QFAC ${\cal M}$  always has $L_{\cal M} (\epsilon)=1$, that is, the  initial  state is an accepting state.

\begin{theorem} \label{QSCT1} Let $\Sigma$ be a finite event set and $\Sigma=\Sigma_{uc}\cup \Sigma_{c}$.
Suppose a QDES with event set $\Sigma$ is modeled as $L_{\cal M}^{\lambda}$ for a 1QFAC ${\cal M}$ with $0\leq \lambda<1$. Let $K\subset \Sigma^*$  be recognized  by a 1QFAC ${\cal H}$ with cut-point $\mu$ ($\mu\geq \lambda$) isolated by $\rho$, where    $pr(K)\subseteq L_{\cal M}^{\lambda}$ and  $L_{\cal H}\leq L_{\cal M}$ (but $L_{\cal H} (x) = L_{\cal M}(x)$ for $x\in pr(K)$ ). If  $\forall s\in\Sigma^{\ast}$, $\forall  \sigma\in \Sigma_{uc}$,
\begin{equation}\label{CD}
\min\{L_{\cal H}(s), L_{\cal M}(s\sigma)\}\leq L_{\cal H}(s\sigma),
\end{equation}
then  there is a quantum supervisor ${\cal S}:  L_{\cal M}^{\lambda} \rightarrow [0,1]^\Sigma$
such that

\begin{equation}
L_{{\cal S}/{\cal M}}^{\mu}\subseteq K\subseteq L_{{\cal S}/{\cal M}}^{\lambda}.
\end{equation}

\end{theorem}

\begin{proof}   
     Let
 \begin{align}
     {\cal S}(s) (\sigma)=
          \begin{cases}
           L_{\cal M}(s\sigma),                 & \text{if }  \sigma\in\Sigma_{uc},\\
           L_{\cal H}(s\sigma),                 & \text{if } \sigma\in\Sigma_{c}.\\
          \end{cases}
        \end{align}

Recall $\forall s\in\Sigma^{\ast}$, $\forall \sigma\in\Sigma$,

 \begin{equation}
 L_{{\cal S}/{\cal M}}(s\sigma)=\min\{L_{{\cal S}/{\cal M}}(s), L_{\cal M}(s\sigma),{\cal S}(s)(\sigma)\}.
 \end{equation}

 Suppose that $L_{\cal H}(s)=L_{{\cal S}/{\cal M}}(s)$ for $|s|\leq n$. Then $\forall \sigma\in\Sigma$,

  \begin{align}
     & L_{{\cal S}/{\cal M}}(s\sigma) =\nonumber\\
          &\begin{cases}
          \min \{ L_{\cal H}(s), L_{\cal M}(s\sigma)\}\\
          \leq L_{\cal H}(s\sigma),          & \text{if }  \sigma\in\Sigma_{uc},\\
          \min\{ L_{\cal H}(s), L_{\cal M}(s\sigma), L_{\cal H}(s\sigma)\}\\
          \leq L_{\cal H}(s\sigma),         & \sigma\in\Sigma_{c}.\\
          \end{cases}
        \end{align}

 On the other hand,

  \begin{align}
      L_{\cal H}(s\sigma) & \leq \min\{L_{\cal M}(s\sigma), L_{\cal H}(s)\}\\
          &=\begin{cases}
          \min \{ L_{\cal M}(s\sigma), S(s)(\sigma), L_{{\cal S}/{\cal M}}(s)\}\\
          = L_{{\cal S}/{\cal M}}(s\sigma),          & \text{if }  \sigma\in\Sigma_{uc},\\
          \min\{ L_{\cal M}(s\sigma), L_{{\cal S}/{\cal M}}(s), L_{\cal H}(s\sigma)\}\\
        =L_{{\cal S}/{\cal M}}(s\sigma),         & \sigma\in\Sigma_{c}.\\
          \end{cases}
        \end{align}
\hskip 115mm $\blacksquare$


\end{proof}

\begin{remark}
Theorem \ref{QSCT1} shows that under certain conditions, there is a quantum supervisor to achieve an approximate objective specification.
Next we give a sufficient and necessary condition for the existence of quantum supervisor to achieve a precise supervisory control, and this is described by the following theorem. \quad $\blacksquare$
\end{remark}

\begin{theorem} \label{QSCT2}
Let $\Sigma$ be a finite event set and $\Sigma=\Sigma_{uc}\cup \Sigma_{c}$.  Suppose a QDES with event set $\Sigma$ is modeled as a quantum language $L_{\cal M}$ that is generated by an  1QFAC ${\cal M}$. Quantum language ${\cal K}$ generated by some 1QFAC satisfies $pr({\cal K})\subseteq L_{\cal M}$.
Then there is a quantum supervisor ${\cal S}:  \Sigma^{\ast} \rightarrow [0,1]^\Sigma$
such that
 $L_{{\cal S}/{\cal M}}=pr({\cal K})$, if and only if  $\forall s\in\Sigma^{\ast}$, $\forall  \sigma\in \Sigma_{uc}$,

\begin{equation}\label{CD}
\min\{pr({\cal K})(s), L_{\cal M}(s\sigma)\}\leq pr({\cal K})(s\sigma).
\end{equation}

\end{theorem}

\begin{proof}   $\Leftarrow$\text{)}.     Let

 \begin{align}
     {\cal S}(s) (\sigma)=
          \begin{cases}
           L_{\cal M}(s\sigma),                 & \text{if }  \sigma\in\Sigma_{uc},\\
           pr({\cal K})(s\sigma),                 &  \text{if } \sigma\in\Sigma_{c}.\\
          \end{cases}
        \end{align}
First $L_{\cal M} (\epsilon)=1$ holds as we suppose the initial state is an accepting state. Recall $\forall s\in\Sigma^{\ast}$, $\forall \sigma\in\Sigma$,

 \begin{equation}
 L_{{\cal S}/{\cal M}}(s\sigma)=\min\{ L_{{\cal S}/{\cal M}}(s), L_{\cal M}(s\sigma),{\cal S}(s)(\sigma)\}.
 \end{equation}

 Suppose that $pr({\cal K})(s)= L_{{\cal S}/{\cal M}}(s)$ for $|s|\leq n$. Then $\forall \sigma\in\Sigma$,

  \begin{align}\label{Th1Eq1}
     & L_{{\cal S}/{\cal M}}(s\sigma) =\nonumber\\
          &\begin{cases}
          \min \{ pr({\cal K})(s), L_{\cal M}(s\sigma)\}\\
          \leq pr({\cal K})(s\sigma),          & \text{if }  \sigma\in\Sigma_{uc},\\
          \min\{ pr({\cal K})(s), L_{\cal M}(s\sigma), pr({\cal K})(s\sigma)\}\\
          \leq pr({\cal K})(s\sigma),         & \text{if } \sigma\in\Sigma_{c}.\\
          \end{cases}
        \end{align}

 On the other hand,

  \begin{align}\label{Th1Eq2}
     & pr({\cal K})(s\sigma) \\
     & \leq \min\{L_{\cal M}(s\sigma), pr({\cal K})(s)\}\\
          &=\begin{cases}
          \min \{ L_{\cal M}(s\sigma), {\cal S}(s)(\sigma),  L_{{\cal S}/{\cal M}}(s)\}\\
          = L_{{\cal S}/{\cal M}}(s\sigma),          & \text{if }  \sigma\in\Sigma_{uc},\\
          \min\{ L_{\cal M}(s\sigma), L_{{\cal S}/{\cal M}}(s), pr({\cal K})(s\sigma)\}\\
        =L_{{\cal S}/{\cal M}}(s\sigma),         & \text{if } \sigma\in\Sigma_{c}.\\
          \end{cases}
        \end{align}

 $\Rightarrow$).  Let quantum supervisor ${\cal S}:  \Sigma^{\ast}\rightarrow [0,1]^{\Sigma}$ satisfy that ${\cal S}(s)(\sigma)\geq L_{\cal M}(s\sigma)$ for any $s\in\Sigma^{\ast}$,  any $\sigma\in\Sigma_{uc}$.

  \begin{align}
     &\min \{pr({\cal K})(s),L_{\cal M}(s\sigma)\}\nonumber\\
 = & \min\{ L_{{\cal S}/{\cal M}}(s),    L_{\cal M}(s\sigma)\}\\
    =&\min\{ L_{{\cal S}/{\cal M}}(s),  L_{\cal M}(s\sigma), {\cal S}(s)(\sigma)\}\\
    =&  L_{{\cal S}/{\cal M}}(s\sigma)\\
    =& pr({\cal K})(s\sigma).
    \end{align}

So, we complete the proof of theorem. 
\quad $\blacksquare$
\end{proof}

\vskip 5mm
 From Theorem \ref{QSCT2} we can obtain a corollary, and this is a modelling fashion of QDES with cut-point.

\begin{corollary} \label{Co1}
Suppose a QDES is modeled as a cut-point language $L_{\cal M}^{\lambda}$ recognized by a  1QFAC ${\cal M}$ with $0\leq\lambda<1$. Quantum language ${\cal K}$ generated by some 1QFAC satisfies $pr({\cal K})\subseteq L_{\cal M}$.
Then there is a quantum supervisor ${\cal S}:  \Sigma^{\ast} \rightarrow [0,1]^\Sigma$
such that
 $L_{{\cal S}/{\cal M}}^{\lambda}=pr({\cal K})^{\lambda}$, if and only if  $\forall s\in pr({\cal K})^{\lambda}$, $\forall  \sigma\in \Sigma_{uc}$, if $s\sigma\in L_{\cal M} ^{\lambda}$, then $s\sigma\in pr({\cal K})^{\lambda}$.

\end{corollary}

\begin{proof}
 $\Leftarrow$).     Let
 \begin{align}
      {\cal S}(s) (\sigma)=
          \begin{cases}
           L_{\cal M}(s\sigma),                 & \text{if }  \sigma\in\Sigma_{uc},\\
           pr({\cal K})(s\sigma),                 &\text{if }  \sigma\in\Sigma_{c}.\\
          \end{cases}
        \end{align}

By means of Inequalities  \ref{Th1Eq1} and \ref{Th1Eq2} we can obtain that $s\in pr({\cal K})^{\lambda}$ if and only if $s\in L_{{\cal S}/{\cal M}}^{\lambda}$, for any $ s\in \Sigma^{\ast}$.

 $\Rightarrow$). Recall the supervisor ${\cal S}$ satisfies the quantum admissible condition (Eq. (\ref{admissible})): for any $s\in\Sigma^{\ast}$,  any $\sigma\in\Sigma_{uc}$, ${\cal S}(s)(\sigma)\geq L_{\cal M}(s\sigma)$.  Therefore, with the definition
 $L_{{\cal S}/{\cal M}}^{\lambda}$,  we have $s\sigma\in pr({\cal K})^{\lambda}$. \quad $\blacksquare$

\end{proof}

\subsection{An example to illustrate supervisory control theorems of QDES}

\textit{Theorem} \ref{QSCT2} and its \textit{Corollary} \ref{Co1} are the main supervisory control results of QDES.  As an application of \textit{Theorem} \ref{QSCT2} (or \textit{Corollary} \ref{Co1}), we present an example to show the advantage of QDES over classical DES in state complexity. 

\begin{example} We employ the language $L(m)$ in Example \ref{EGnew}.
 Let $\Sigma=\{0,1\}$,  $\Sigma_{uc}=\{0\}$, $\Sigma_{c}=\{1\}$. Given an natural number $m>2$, suppose a QDES modeled as the cut-point language \begin{equation}
L_{\cal M}^{0}=L(m)=\{w\in\{0,1\}^*:0^{km}1 \text{ is not a prefix of } w,\forall k\in \mathbb{N} \}
\end{equation}
 recognized by a 1QFAC $ {\cal M} $ with cut-point $0$, and quantum language ${\cal K}$ is generated by another 1QFAC ${\cal M}_{{\cal K}}$ that will be defined in the following.

 ${\cal M}$ is defined as \textit{Example} \ref{EGnew}, and  $L_{{\cal M}}(w)>0$ for any $w\in \Sigma^{\ast}$ iff   $w\in L(m)$.


For any given natural number $q>2$, we consider the language $L(mq)$ that is recognized by  another 1QFAC 
 $ {\cal M}_{\cal K}$ with cut-point $0$.  $ {\cal M}_{\cal K}$
can be defined by means of $ {\cal M}$, and ${\cal K}(w)=L_{{\cal M}_{\cal K}}(w)>0$ for any $w\in \Sigma^{\ast}$ iff   $w\in L(mq)$.

 Therefore, we have
 $L_ {\cal M}^{0}=L(m)$ and $pr({\cal K})^{0}={\cal K}^{0}=L(mq)$.

It is  immediate to check that the condition ``$\forall w\in pr({\cal K})^{0}$, $\forall  \sigma\in \Sigma_{uc}$, if $w\sigma\in L_{\cal M} ^{0}$, then $w\sigma\in pr({\cal K})^{0}$" in
 Corollary \ref{Co1} holds, so  there is a quantum supervisor $S:  \Sigma^{\ast} \rightarrow [0,1]^\Sigma$
such that
 $L_{{\cal S}/{\cal M}}^{0}=pr({\cal K})^{0}$.

Due to Theorem \ref{EGtheorem},
we know that PFA (i.e., PDES) require $\lceil\log_2m\rceil$ and $\lceil\log_2mq\rceil$ states to recognize the languages $L_{\cal M}^{0}$ and $pr({\cal K})^{0}$ with cut-point $0$, respectively. So, QDES show essential advantage over classical DES in state complexity for simulation of systems.


\end{example}
\hskip 123mm $\blacksquare$

\vskip 10mm

\subsection{Supervisory Control of QDES simulated with isolated  cut-point  languages}

In this section we study {\it nonblocking} problem in QDES. In this case, it is more suitsble to consider QDES  simulated with the languages (say $L$) of cut-point $\lambda$ isolated by an $\rho$, and the controlled language by quantum supervisor ${\cal S}$ is required to belong to this language $L$. Our main purpose is to establish {\it nonblocking} quantum supervisor theorem in QDES.

First, we would  recall {\it nonblocking} problem in classical DES \cite{book1, book2}. Let a DES modeled by finite automaton $G=(Q,\Sigma,\delta,q_{0},Q_{m})$. Then $G$ is called {\it nonblocking} if $pr(L_m(G))=L(G)$. That is, for any feasible input string $s$, input string $st$ will reach a marked state for some string $t$. Let $S$ be a supervisor. Then the language marked by $S/G$ is defined as:
\begin{equation}
L_{m}(S/G)=L(S/G)\cap L_{m}(G).
\end{equation}
The DES modeled by $L(S/G)$ is {\it nonblocking} if $L(S/G)= pr
(L_{m}(S/G)).$ 
In practice, $L(S/G)$ being {\it nonblocking} is important since it leads to each string in $L(S/G)$  together with some more input string being able to reach a marked state.

Suppose that  1QFAC ${\cal M}$ recognizes a language (denoted by $L_{\cal M}^{\lambda,\rho}$) over alphabet $\Sigma$ with cut-point $\lambda$ isolated by $\rho$. For any $s\in\Sigma^{\ast}$, denote

\begin{align}
          L_{{\cal M},a}(s)=
        \begin{cases}
           L_{\cal M}(s),
            & \text{if }   s\in L_{\cal M}^{\lambda,\rho},\\
        0,         & \text{otherwise},\\
          \end{cases}
 \end{align}

and

\begin{align}
          L_{{\cal S}/{\cal M},a}(s)=
        \begin{cases}
          \min \{ L_{\cal M}(s), L_{{\cal S}/{\cal M}}(s)\},
            & \text{if }  s\in L_{\cal M}^{\lambda,\rho},\\
        0,         & \text{otherwise}.\\
          \end{cases}
 \end{align}

 For any quantum language $L$ over $\Sigma$, denote by $\text{supp} (L)$ the support set of $L$, i.e., $\text{supp} (L)=\{s\in\Sigma^{\ast}: L(s)> 0\}$. A quantum supervisor ${\cal S}$ is called as \textit{nonblocking} if it satisfies $L_{{\cal S}/{\cal M}}=pr(L_{{\cal S}/{\cal M},a})$.

\begin{theorem} \label{MARKINGQSCT} Suppose a QDES is modeled as a language $L_{{\cal M},a}$  recognized by a 1QFAC ${\cal M}$ with cut-point $\lambda$ isolated by $\rho$. ${\cal K}$ is a quantum language over $\Sigma$.
Let $pr({\cal K})\leq L_{{\cal M},a}$. 
Then there is a quantum supervisor ${\cal S}$ satisfying nonblocking (i.e., $L_{{\cal S}/{\cal M}}=pr(L_{{\cal S}/{\cal M},a})$) such that ${\cal K}=L_{{\cal S}/{\cal M},a}$ and $pr({\cal K})=L_{{\cal S}/{\cal M}}$ if and only if

\begin{enumerate}
  \item $\forall s\in\Sigma^{\ast}$, $\forall \sigma\in\Sigma_{uc}$,
 \begin{align}
\min \{pr({\cal K})(s),L_{\cal M}(s\sigma)\}\leq pr({\cal K})(s\sigma);
\end{align}
  \item $\forall s\in\Sigma^{\ast}$,
  \begin{align}
{\cal K}(s)=\min\{pr({\cal K})(s),L_{{\cal M},a}(s)\}.
  \end{align}
\end{enumerate}

\end{theorem}

\begin{proof}
$\Leftarrow$).  $\forall s\in\Sigma^{\ast}$, let

\begin{align}
          {\cal S}(s)(\sigma)=
        \begin{cases}
                   L_{\cal M}(s\sigma), & \sigma\in\Sigma_{uc}, \\
                   pr({\cal K})(s\sigma), & \sigma\in\Sigma_{c}.\\
                 \end{cases}
\end{align}

First, $L_{{\cal S}/{\cal M}}(\epsilon)=1=pr({\cal K})(\epsilon)=\sup_{t\in\Sigma^{\ast}}{\cal K}(t)=1$ (${\cal K}(\epsilon)=1$). Suppose $s\in\Sigma^{\ast}$ and $|s|\leq n$, $L_{{\cal S}/{\cal M}}(s)=pr({\cal K})(s)$. Then $\forall \sigma\in\Sigma$,

(I) if $\sigma\in\Sigma_{uc}$, then

\begin{align}
  L_{{\cal S}/{\cal M}}(s\sigma) = & \min \{ L_{{\cal S}/{\cal M}}(s), L_{\cal M}(s\sigma),{\cal S}(s)(\sigma)\} \\
   =& \min\{L_{{\cal S}/{\cal M}}(s), L_{{\cal M}}(s\sigma)\}\\
   =& \min \{ pr({\cal K})(s),  L_{{\cal M}}(s\sigma)\}\\
   \leq & pr({\cal K})(s\sigma);
\end{align}

(II) if $\sigma\in\Sigma_{c}$, then it holds as well, since ${\cal S}(s)(\sigma)= pr({\cal K})(s\sigma)$.
On the other hand,
\begin{align}
 pr({\cal K})(s\sigma) \leq & \min \{ pr({\cal K})(s), L_{\cal M}(s\sigma)\} \\
   =&   \min\{pr({\cal K})(s),  L_{\cal M}(s\sigma), {\cal S}(s)(\sigma)\}  \\
   =& L_{{\cal S}/{\cal M}}(s\sigma).
\end{align}

So, $pr({\cal K})=L_{{\cal S}/{\cal M}}$.

For any $s\in\Sigma^{\ast}$, if $L_{\cal M}(s)<\lambda+\rho$, then $ L_{{\cal S}/{\cal M},a}(s)={\cal K}(s)=0$; if $L_{\cal M}(s)\geq\lambda+\rho$, then with condition 2) above, we have
\begin{align}
{\cal K}(s) = & \min \{ pr({\cal K})(s), L_{{\cal M},a}(s)\} \\
   =&   \min\{pr({\cal K})(s),  L_{\cal M}(s)\}  \\
   =& \min \{ L_{{\cal S}/{\cal M}}(s),  L_{{\cal M}}(s)\}\\
   =& L_{{\cal S}/{\cal M},a}(s).
\end{align}

$\Rightarrow$). 1) $\forall s\in\Sigma^{\ast}$, $\forall \sigma\in\Sigma_{uc}$, since quantum supervisor ${\cal S}$ always satisfies  that $L_{\cal M}(s\sigma)\leq {\cal S}(s)(\sigma)$, we have
\begin{align}
 & \min \{ pr({\cal K})(s), L_{{\cal M}}(s\sigma)\} \\
   \leq &   \min\{L_{{\cal S}/{\cal M}}(s), L_{\cal M}(s\sigma),{\cal S}(s)(\sigma)\}  \\
   =& L_{{\cal S}/{\cal M}}(s\sigma).
\end{align}

2)  $\forall s\in\Sigma^{\ast}$, if $L_{\cal M}(s)<\lambda+\rho$, then ${\cal K}(s)=0$ and $pr({\cal K})(s)=L_{{\cal M},a}(s)=0$; if $L_{\cal M}(s)\geq\lambda+\rho$, then
\begin{align}
{\cal K}(s) =& L_{{\cal S}/{\cal M},a}(s) \\
   =&   \min\{L_{{\cal S}/{\cal M}}(s), L_{\cal M}(s)\}  \\
   =&  \min\{pr({\cal K})(s), L_{{\cal M},a}(s)\}.
\end{align}

Consequently, the proof is completed. \quad $\blacksquare$

\end{proof}

As a special case of Theorem  \ref{MARKINGQSCT}, the following corollary follows.

\begin{corollary}  \label{MARKINGQSCTCOR}
Suppose a QDES is modeled as a language $L_{{\cal M}}^{\lambda,\rho}$ recognized by a 1QFAC ${\cal M}$ with cut point $\lambda$ isolated by $\rho$. $K\subset\Sigma^*$ is a regular language.
Let $pr(K)\subseteq L_{{\cal M}}^{\lambda,\rho}$. Then there is a quantum supervisor ${\cal S}$ satisfying $L_{{\cal S}/{\cal M}}=pr(L_{{\cal S}/{\cal M},a})$ such that $K=L_{{\cal S}/{\cal M},a}^{0}$ and $pr(K)=L_{{\cal S}/{\cal M}}^0$ if and only if the following two conditions hold:

\begin{enumerate}
  \item
  \begin{align}
pr(K) \Sigma_{uc} \cap  L_{\cal M}^0  \subseteq pr(K);
  \end{align}

  \item
  \begin{align}
K=pr(K)\cap L_{{\cal M}}^{\lambda,\rho},
  \end{align}
\end{enumerate}
where $pr(K) \Sigma_{uc}=\{s\sigma: s\in pr(K), \sigma\in\Sigma_{uc}\}$.

\end{corollary}

\begin{proof}
In fact, we can do it by taking ${\cal K}$ as a classical language in {\it Theorem}  \ref{MARKINGQSCT}.  So, we omit the details here. \quad $\blacksquare$




\end{proof}

\section{Decidability of Controllability Condition} \label{SECDCC}

In supervisory control of QDES, the  controllability conditions play an important role of the existence of quantum supervisors. So,
we present a polynomial-time algorithm  to decide the controllability condition Eq. (\ref{CD}). The prefix-closure of quantum language ${\cal K}$, as  the target language we hope to achieve under the supervisory control,  is in general generated by a 1QFAC ${\cal H}$,  that is, $pr({\cal K})=L_{{\cal H}}$. Then the controllability condition Eq. (\ref{CD}) is equivalently as:
$\forall s\in\Sigma^{\ast}$, $\forall  \sigma\in \Sigma_{uc}$,
\begin{equation}\label{CD1}
\min\{L_{\cal H}(s), L_{\cal M}(s\sigma)\}\leq L_{\cal H}(s\sigma).
\end{equation}

First, we need a proposition.
\begin{proposition}\label{EQ}

Suppose a QDES modeled as a quantum language $L_{\cal M}$ that is generated by a 1QFAC ${\cal M}$. Quantum language ${\cal K}$ satisfies $pr({\cal K})\subseteq L_{\cal M}$, and  $pr({\cal K})=L_{{\cal H}}$ for  some  1QFAC  ${\cal H}$.
Then  $L_{\cal H}(s)\geq L_{\cal H}(s\sigma)$ and $L_{\cal M}(s\sigma)\geq L_{\cal H}(s\sigma)$ for any  $s\in\Sigma^{\ast}$, and $ \sigma\in \Sigma_{uc}$.

\end{proposition}

\begin{proof} By the definition of  prefix closure, for any  $s\in\Sigma^{\ast}$, and $ \sigma\in \Sigma_{uc}$,  we have  $L_{\cal H}(s)=pr({\cal K})(s)=\sup_{t\in\Sigma^{\ast}}{\cal K}(st)$ and  $ L_{\cal H}(s\sigma)= pr({\cal K})(s\sigma)=\sup_{t\in\Sigma^{\ast}}{\cal K}(s\sigma t)$, which result in  $L_{\cal H}(s)\geq L_{\cal H}(s\sigma)$.

Furthermore, due to $pr({\cal K})\subseteq L_{\cal M}$ and $pr({\cal K})=L_{{\cal H}}$, we have $ L_{{\cal H}}  \subseteq L_{\cal M}$, that is, for any  $s\in\Sigma^{\ast}$, and $ \sigma\in \Sigma_{uc}$, $  L_{\cal H}(s\sigma)\leq    L_{\cal M}(s\sigma)$.

\end{proof}

In the light of Proposition \ref{EQ}, we have:

\begin{align}
& \min\{L_{\cal H}(s),L_{\cal M}(s\sigma)\}\leq L_{{\cal H}}(s\sigma) \\
  \Leftrightarrow &   \min\{L_{\cal H}(s),  L_{\cal M}(s\sigma)\}= L_{{\cal H}}(s\sigma) \\
    \Leftrightarrow &   \frac{L_H(s)+L_{\cal M}(s\sigma)-|L_{\cal H}(s)-L_{\cal M}(s\sigma)|}{2}= L_{{\cal H}}(s\sigma)\\
    \Leftrightarrow &   L_{\cal H}(s)+L_{\cal M}(s\sigma)-2L_{\cal H}(s\sigma)\\
    &= |L_{{\cal H}}(s)-L_{\cal M}(s\sigma)|\\
    \Leftrightarrow &   (L_{\cal H}(s)+L_{\cal M}(s\sigma)-2L_{\cal H}(s\sigma))^{2}\\
    &= (L_{{\cal H}}(s)-L_{\cal M}(s\sigma))^2\\
     \Leftrightarrow &   L_{\cal H}(s)L_{\cal M}(s\sigma)+L_{\cal H}(s\sigma)^2\\
     &=L_{\cal H}(s\sigma)L_{\cal H}(s)+L_{\cal H}(s\sigma)L_{\cal M}(s\sigma).  \label{ECD1}
\end{align}
So, Inequality (\ref{CD1}) is equivalent to Eq. (\ref{ECD1}), and therefore it suffices to check whether Eq. (\ref{ECD1}) holds for any $ s\in\Sigma^{\ast}$, and for each $ \sigma\in \Sigma_{uc}$, in order to check the controllability condition.

In fact, we have the following result, where $|Q_{\cal M}|$ and $|Q_{\cal H}|$  are the numbers of quantum basis states  of ${\cal M}$ and ${\cal H}$, respectively.

\begin{theorem}
Suppose a QDES modeled as a quantum language $L_{\cal M}$ generated by a 1QFAC ${\cal M}$. For quantum language ${\cal K}$, $pr({\cal K})$ is generated by another  1QFAC ${\cal H}$ and  $pr({\cal K})\subseteq L_{\cal M}$.
Then the controllability condition Eq. (\ref{CD}) holds if and only if for any $ \sigma\in \Sigma_{uc}$,
 for any $s\in\Sigma^{\ast}$ with $|s|\leq 18|Q_{\cal H}|^2(|Q_{\cal H}|^2+|Q_{\cal M}|^2)-1$, Eq. (\ref{CD1}) holds. Furthermore, there exists a polynomial-time algorithm running in
time $O(|\Sigma||Q_{\cal H}|^6(|Q_{\cal H}|^2+|Q_{\cal M}|^2)^3)$ that determines whether the controllability condition Eq. (\ref{CD}) holds.

\end{theorem}

\begin{proof}

According to Lemma \ref{1QFACTOBLM}, 1QFAC ${\cal M}$ and ${\cal H}$ can be simulated by two RBLM, say $M$ and $H$ respectively, such that $\forall s\in\Sigma^{\ast}$,
\begin{equation}
L_\mathcal{M}(s) = f_{M}(s),
\end{equation}
\begin{equation}
L_\mathcal{H}(s) = f_{H}(s),
\end{equation}
and the numbers of states in $M$ and $H$ are $ |S_{\cal M}|^2   |Q_{\cal M}|^2     $ and $|S_{\cal H}|^2 |Q_{\cal H}|^2        $, respectively,  where     $|S_{\cal M}|$ and $|S_{\cal H}|$ represent respectively the numbers of classical states in  ${\cal M}$ and ${\cal H}$,       $|Q_{\cal M}|$ and $|Q_{\cal H}|$ represent respectively the numbers of quantum states in  ${\cal M}$ and ${\cal H}$, and functions $f_M$ and $f_H$ are associated to $M$ and $H$, respectively.

Similarly, by virtue of Proposition \ref{BLMTOBLM} and Lemma \ref{1QFACTOBLM}, for each $\sigma\in \Sigma_{uc}$, there exist two RBLM $M_{\sigma}$ and $H_{\sigma}$ respectively satisfying that $\forall s\in\Sigma^{\ast}$,

\begin{enumerate}
\item \begin{equation}
L_{\cal M}(s\sigma) = f_{M_{\sigma}}(s),
\end{equation}
\item \begin{equation}
L_{\cal H}(s\sigma) = f_{H_{\sigma}}(s),
\end{equation}
\end{enumerate}
where the numbers of states in $M_{\sigma}$ and $H_{\sigma}$ are also $     |S_{\cal M}|^2       |Q_{\cal M}|^2      $ and $     |S_{\cal H}|^2        |Q_{\cal H}|^2     $, respectively. Therefore, $\forall s\in\Sigma^{\ast}$,

\begin{enumerate}
\item \begin{equation}
L_{\cal H}(s)L_{\cal M}(s\sigma) = f_{H}(s)f_{M_{\sigma}}(s)=f_{H\otimes M_{\sigma}}(s),
\end{equation}
\item
\begin{equation}
L_{\cal H}(s\sigma)^2 = f_{H_{\sigma}}(s)f_{H_{\sigma}}(s)=f_{H_{\sigma}\otimes H_{\sigma}}(s),
\end{equation}
\item
\begin{equation}
L_{\cal H}(s\sigma)L_{\cal H}(s) =f_{H_{\sigma}}(s)f_{H}(s)=f_{H_{\sigma}\otimes H}(s),
\end{equation}
\item
\begin{equation}
L_{\cal H}(s\sigma)L_{\cal M}(s\sigma) = f_{H_{\sigma}}(s)f_{M_{\sigma}}(s)=f_{H_{\sigma}\otimes M_{\sigma}}(s),
\end{equation}
\end{enumerate}
where the second equalities of each equations above are due to  Remark \ref{TPRBLM}. Therefore, equation (\ref{ECD1}) is equivalent to
\begin{equation}
f_{H\otimes M_{\sigma}}(s)+f_{H_{\sigma}\otimes H_{\sigma}}(s)=f_{H_{\sigma}\otimes M_{\sigma}}(s)+f_{H_{\sigma}\otimes H}(s)
\end{equation}
for every $ s\in\Sigma^{\ast}$. Furthermore, by means of Remark \ref{TPRBLM}, we have
\begin{equation} \label{ECD2}
f_{(H\otimes M_{\sigma})\oplus (H_{\sigma}\otimes H_{\sigma})}(s)=f_{(H_{\sigma}\otimes M_{\sigma})\oplus (H_{\sigma}\otimes H)}(s)
\end{equation}
for every $s\in\Sigma^{\ast}$, where the state numbers of $(H\otimes M_{\sigma})\oplus (H_{\sigma}\otimes H_{\sigma})$ and $(H_{\sigma}\otimes M_{\sigma})\oplus (H_{\sigma}\otimes H)$ are the same as
\begin{align}
&      |S_{\cal H}|^2|Q_{\cal H}|^2      |S_{\cal M}|^2|Q_{\cal M}|^2    +     |S_{\cal H}|^4\       |Q_{\cal H}|^4\\
=&      |S_{\cal H}|^2|Q_{\cal H}|^2       (   |S_{\cal M}|^2|Q_{\cal M}|^2+     |S_{\cal H}|^2|Q_{\cal H}|^2  ).
\end{align}

By virtue of Proposition \ref{BLMEQ}, the above Equation  (\ref{ECD2}) holds for every $s\in\Sigma^{\ast}$ if and only if it holds for all $s\in\Sigma^{\ast}$ with
$|s|\leq    2  |S_{\cal H}|^2|Q_{\cal H}|^2       (   |S_{\cal M}|^2|Q_{\cal M}|^2+     |S_{\cal H}|^2|Q_{\cal H}|^2  )   $, and there  exists a polynomial-time algorithm running in
time $O(|S_{\cal H}|^6|Q_{\cal H}|^6      (   |S_{\cal M}|^2|Q_{\cal M}|^2+     |S_{\cal H}|^2|Q_{\cal H}|^2  ) ^3$ to determine whether  Equation  (\ref{ECD2}) holds for every $s\in\Sigma^{\ast}$. Here we present the algorithm in detail, but omit the analyses of correctness and complexity and the details are referred to \cite{QLZMG11,LQ08,CMR2006,KMOWW2011,LF2015,WLY2021}.

In the first step,  given 1QFAC ${\cal M}$ and ${\cal H}$, and for any $\sigma\in\Sigma_{uc}$,   we can directly construct  two RBLM $(H\otimes M_{\sigma})\oplus (H_{\sigma}\otimes H_{\sigma})$ and $(H_{\sigma}\otimes M_{\sigma})\oplus (H_{\sigma}\otimes H)$ as above, and for simplicity,  we denote them respectively by
\begin{itemize}
\item ${\cal M}_1(\sigma)=(S_1, \pi_1, \{M_1(t)\}_{t\in\Sigma},\eta_1)$,
\item ${\cal M}_2(\sigma)=(S_2, \pi_2, \{M_2(t)\}_{t\in\Sigma},\eta_2)$.
\end{itemize}

 Recalling Definition \ref{BLM} and Remark \ref{TPRBLM}, we have
 \begin{equation}
 f_{{\cal M}_i(\sigma)}(w)=\eta_i M_i(w_m)M_i(w_{m-1})\ldots M_i(w_1)\pi_i,
 \end{equation}
 $i=1,2,$
 their direct sum  is
  \begin{align}
  &{\cal M}_1(\sigma)\oplus {\cal M}_2(\sigma)\\
  =&(S_1\oplus S_2, \pi_1\oplus \pi_2, \{M_1(t)\oplus M_2(t)\}_{t\in\Sigma},\eta_1\oplus \eta_2),
   \end{align}
  and then
   \begin{equation}
f_{\mathcal{M}_1\oplus \mathcal{M}_2}(w)=f_{\mathcal{M}_1}(w) + f_{\mathcal{M}_2}(w)
 \end{equation}
for any $w\in\Sigma^*$. For any $w=w_1w_2\ldots w_m\in \Sigma^{\ast}$,
denote
 \begin{align}
P_{{\cal M}_i(\sigma)}(w)
= M_i(w_m)M_i(w_{m-1})\ldots M_i(w_1)\pi_i,
 \end{align}
 for $i=1,2.$



Now we present Algorithm \ref{alg:1} to check whether or not $M_1(\sigma)$ and $M_2(\sigma)$ are equivalent.

So, if for any $ \sigma\in \Sigma_{uc}$, Algorithm \ref{alg:1} returns $M_1(\sigma)$ and $M_2(\sigma)$ are equivalent, then the controllability condition holds; otherwise the controllability condition does not hold.

\begin{algorithm}
	\renewcommand{\algorithmicrequire}{\textbf{Input:}}
	\caption{Algorithm for checking the equivalence between ${\cal M}_1(\sigma)$ and ${\cal M}_2(\sigma)$}
	\label{alg:1}
	\begin{algorithmic}[1]
		\REQUIRE RBLM ${\cal M}_1(\sigma)=(S_1, \pi_1, \{M_1(t)\}_{t\in\Sigma},\eta_1)$ and ${\cal M}_2(\sigma)=(S_2, \pi_2, \{M_2(t)\}_{t\in\Sigma},\eta_2)$;
		\STATE Set $\textbf{V}$ and $\textbf{N}$ to be the empty set;
		\STATE queue $\leftarrow$  node($\varepsilon$); $\varepsilon$ denotes empty string;
		\WHILE{queue is not empty}
             \STATE \textbf{begin} take an element node($x$) from queue; $x\in\Sigma^*$;
             \IF {$P_{{\cal M}_1(\sigma)\oplus {\cal M}_2(\sigma)}(x)\notin span(\textbf{V})$}
		\STATE \textbf{begin} add all node($xt$) for $t\in\Sigma$ to queue;
             \STATE \qquad\quad add vector $P_{{\cal M}_1(\sigma)\oplus {\cal M}_2(\sigma)}(x)$ to $\textbf{V}$;
             \STATE \qquad\quad add node($x$) to $\textbf{N}$;
             \STATE \textbf{end};
             \ENDIF;
             \STATE \textbf{end};
             \ENDWHILE;
		\IF {$\forall\textbf{v}\in \textbf{V}$, $(\eta_{1}\oplus -\eta_{2})\textbf{v}=0$}
             \STATE return(yes)
             \ELSE
             \STATE return an $x_0\in \{x:node(x)\in\textbf{N}, (\eta_{1}\oplus -\eta_{2})P_{{\cal M}_1(\sigma)\oplus {\cal M}_2(\sigma)}(x)\neq0\}$;
             \ENDIF;
	\end{algorithmic}
\end{algorithm}

\end{proof}

\section{Concluding Remarks} \label{SECCR}

As a kind of important control systems, DES have been developed deeply  \cite{book1}-\cite{Qiu05} due to the  potential of practical application, but state complexity is still a key problem in DES to be solved appropriately. In recent thirty years, quantum computing has been studied rapidly \cite{NC00}, and quantum control has attracted great interest \cite{WM10}.  So, initiating the study of QDES likely becomes a new subject of DES, and it is also motivated by two aspects: one is the simulation of DES in quantum systems by virtue of the principle of quantum computing; another  is that QDES have advantages over classical DES for processing some problems in state complexity. This paper has been the first study for establishing QDES in light of QFA.

In this paper, we have established a basic framework of QDES, and the supervisory control of QDES has been studied. The main contributions are: (1) We have used 1QFAC to simulate QDES, and proved that MO-1QFA are not suitable for modeling QDES since we found MO-1QFA cannot recognize any prefix-closured language even with cut-point. (2) We have established a number of supervisory control theorems of QDES and proved the sufficient and necessary conditions of the existence of quantum supervisors. (3) We have constructed a number of examples to illustrate the supervisory control theorems and these examples have also showed the advantages of QDES over classical DES concerning state complexity. (4) We have given a polynomial-time algorithm to determine whether or not the controllability condition holds.

In subsequent study, we would like to consider controllability and observability problem of QDES under partial observation of events, and decentralized control of QDES with multi-supervisors under partial observation of events, as well as  diagnosability of QDES. In particular, we would like to discover more practical examples to illustrate the advantages of QDES over classical DES in state complexity.

\section*{Acknowledgements} 
 I thank Ligang Xiao for useful discussion and  helpful construction concerning the examples of QFA recognizing prefix-closured languages.
 This work is supported in part by the National
Natural Science Foundation of China (No. 61876195).





\begin{thebibliography}{0}

\bibitem{book1}
C.~G. Cassandras, S.~Lafortune, \emph{Introduction to Discrete Event Systems}. Springer: New York, 2nd edition, 2008.



\bibitem{book2}
W. M. Wonham, K. Cai, \emph{Supervisory Control of Discrete-Event Systems}. Cham, Switzerland: Springer, 2019.



\bibitem{RW87}
R.~J. Ramadge, W.~M. Wonham,``Supervisory control of a class of discrete event processes,''\emph{SIAM Journal on Control and Optimization}, vol.~25, no.~1, pp. 206--230, 1987.


\bibitem{Kornyak13} V. V. Kornyak,``Classical and quantum discrete dynamical systems," \emph{Physics of Particles and Nuclei}, vol.~44, pp. 47--91, 2013.


\bibitem{Lin90}
F.~Lin,``Supervisory control of stochastic discrete event systems,'' in \emph{Book of Abstracts, SIAM Conference Control 1990's}, San Francisco, 1990.


\bibitem{LW1993}
M.~Lawford, W.~M. Wonham,``Supervisory control of probabilistic discrete
event systems,'' in \emph{Proceedings of the 36th Midwest Symposium on Circuits and Systems}, Detroit, MI, USA, Aug. 1993, pp. 327--331.



\bibitem{PPL2009}
V.~Pantelic, S.~Postma, M.~Lawford,``Probabilistic supervisory control of
probabilistic discrete event systems,''\emph{IEEE Transactions on Automatic Control}, vol.~54, no.~8, pp. 2013--2018, 2009.



\bibitem{LY02}
F.~Lin, H.~Ying,``Modeling and control of fuzzy discrete event systems,'' \emph{IEEE Transactions on Systems, Man, and Cybernetics, Part B (Cybernetics)}, vol.~32, no.~4, pp. 408--415, 2002.



\bibitem{Qiu05}
D. W.~Qiu, ``Supervisory control of fuzzy discrete event systems: a formal approach,'' \emph{IEEE Transactions on Systems, Man, and Cybernetics, Part B (Cybernetics)}, vol.~35, no.~1, pp. 72--88, 2005.








\bibitem{DYQ19}W. Deng, J. Yang, D.W. Qiu, ``Supervisory control of probabilistic discrete event systems under partial observation,'' \emph{IEEE Transactions on Automatic Control}, vol.~64, no.~12, pp. 5051--5065, 2019.





\bibitem{NC00}M.A. Nielsen, I.L. Chuang, \emph{Quantum Computation and Quantum Information}. Cambridge University Press, Cambridge, 2000.



\bibitem{Benioff80} P. Benioff, ``The computer as a physical system: a microscopic quantum mechanical Hamiltonian model of computers as represented by Turing machines,'' \emph{Journal of Statistical Physics}, vol.~22, pp. 563--591, 1980.



\bibitem{Feynman82}R. P. Feynman, ``Simulting physics with computers,'' \emph{International Journal of Theoretical Physics}, vol.~21, pp. 467--488, 1982.



\bibitem{Deutsch85}D. Deutsch, ``Quantum theory, the Church-Turing principle and the universal quantum computer,'' \emph{Proceedings of the Royal Society of London. Series A}, vol.~400, no.~1818, pp. 97--117, 1985.



\bibitem{Shor94}P. W. Shor, ``Polynomial-time algorithms for prime factorization and discrete logarithms on a quantum computer,'' \emph{SIAM Journal on Computing}, vol.~26, no.~5, pp. 1484--1509, 1997.



\bibitem{NJP09}H. I. Nurdina, M. R. James, I. R. Petersen, ``Coherent quantum LQG control,'' \emph{Automatica}, vol.~45, pp. 1837--1846, 2009.



\bibitem{WM10}H. M. Wiseman, G. J. Milburn, \emph{Quantum Measurement and Control}. Cambridge University Press, Cambridge, UK, 2010.



\bibitem{Lloyd2000}
S. Lloyd,``Coherent quantum feedback,"\emph{Physical Review A}, vol. 62, pp. 022108, 2000.


\bibitem{NWCL2000}
Richard J. Nelson, Y. Weinstein, D. Cory, S. Lloyd,``Experimental Demonstration of Fully Coherent Quantum Feedback," \emph{Physical Review Letters},  vol. 85, no. 14, pp. 3045-3048, 2000.





\bibitem{PC05}
Y. Pencol\'{e}, M.-O. Cordier,``A formal framework for the decentralised diagnosis of large scale discrete event systems and its application to telecommunication networks,'' \emph{Artificial Intelligence}, vol.~164, pp. 121--170, 2005.



\bibitem{AF98}
A.~Ambainis, R.~Freivalds,``One-way quantum finite automata: strengths, weaknesses and generalizations,'' in:\emph{Proceedings of the 39th IEEE Symposium on Foundations of Computer Science}, 1998, pp. 332--341.



\bibitem{AY2021}
A. Ambainis, A. Yakaryilmaz, ``Automata and quantum computing," in: \emph{Handbook of Automata Theory (II)}, pp. 1457-1493,2021. arXiv:1507.01988.



\bibitem{MP20}
C. Mereghetti, B. Palano, S. Cialdi, et al., ``Photonic Realization of a Quantum Finite Automaton,'' \emph{Physical Review Research}, vol.~2, 2020, Art. no. 013089.



\bibitem{PHYF2022}
S. Z. D. Plachta, M. Hiekkamäki, A. Yakaryilmaz, R. Fickler, ``Quantum advantage using high-dimensional twisted photons as quantum finite automata," \emph{Quantum}, vol. 6, 752, 2022.



\bibitem{TFLZZ2019} Y. Tian, T. Feng, M. Luo, S. Zheng, and X. Zhou, ``Experimental demonstration of quantum finite
automaton," \emph{ npj Quantum Information}, vol.5, no. 1, pp. 1-5, 2019.


\bibitem{NY2009}H. Nishimura and T. Yamakami, ``An application of quantum finite automata to interactive proof systems," \emph{ Journal of Computer and System Sciences}, vol. 75, no.4, pp. 255-269, 2009.

\bibitem{BZ2020} A.S. Bhatia and S. Zheng,  ``A quantum finite automata approach to modeling the chemical reactions," \emph{Frontiers in Physics}, vol. 8, 547370, 2020.













\bibitem{QLMG12}
D. W. Qiu, L. Li, P. Mateus, J. Gruska, ``Quantum Finite Automata,'' in: \emph{J. Wang (Eds.), Finite State-Based Models and Applications, CRC Handbook}, Boca Raton, pp. 113--144, 2012.



\bibitem{Gruska99}J. Gruska, \emph{Quantum Computing}. McGraw-Hill, London, 1999.



\bibitem{BK19}
A.S. Bhatia, A. Kumar, ``Quantum finite automata: survey, status and research directions,'' 2019, arXiv:1901.07992.


\bibitem{MC00}
C. Moore, J. P. Crutchfield, ``Quantum automata and quantum grammars,'' \emph{Theoretical Computer Science}, vol.~237, pp. 275--306, 2000.



\bibitem{KW97}
A. Kondacs, J. Watrous, ``On the power of finite state automata,'' in \emph{Proceedings of the 38th IEEE Symposium on Foundations of Computer Science}, 1997, pp. 66--75.



\bibitem{BP02}
A. Brodsky, N. Pippenger, ``Characterizations of 1-way quantum finite automata," \emph{SIAM Journal on Computing}, vol.~31, no.~5, pp. 1456--1478, 2002.


\bibitem{MQL12}
P. Mateus, D. W. Qiu, L. Li,``On the complexity of minimizing probabilistic and quantum automata,''\emph{Information and Computation}, vol.~218, pp. 36--53, 2012.

\bibitem{QLZMG11}
D. W. Qiu, L. Li, X. Zou, P. Mateus, J. Gruska,``Multi-letter quantum finite automata: decidability of the equivalence and minimization of states,''\emph{Acta Informatica}, vol.~48, pp. 271--290, 2011.


\bibitem{LQ08}
L. Z. Li, D. W. Qiu,``Determining the equivalence for one-way quantum finite automata,'' \emph{Theoretical Computer Science}, vol.~403, pp. 42--51, 2008.









\bibitem{QLMS15}D.W. Qiu, L. Li, P. Mateus, and A. Sernadas,``Exponentially more concise quantum recognition of non-RMM languages,''\emph{Journal of Computer and System Sciences}, vol.~81, no.~2, pp. 359--375, 2015.



\bibitem{HU79}
J. E. Hopcroft and J. D. Ullman, \emph{Introduction to Automata Theory, Languages, and Computation}. Addison-Wesley, New York, 1979.



\bibitem{Paz71}
A. Paz, \emph{Introduction to Probabilistic Automata}. Academic Press, New York, 1971.



\bibitem{CMR2006}
C. Cortes, M. Mohri and A. Rastogi,``On the computation of some standard distances between probabilistic automata," In: \emph{Proceedings of the 11th International Conference on Implementation and Application of Automata}, 2006, pp. 137-149.



\bibitem{KMOWW2011}
S. Kiefer, A.S. Murawski, J. Ouaknine, B. Wachter and J. Worrell, ``Language equivalence for probabilistic automata," In:\emph{Proceedings of the 23rd International Conference on Computer Aided Verification}, 2011, pp. 526-540.



\bibitem{LF2015}
L. Li and Y. Feng, ``Quantum Markov chains: description of hybrid systems, decidability of equivalence, and model checking linear-time properties," \emph{Information and Computation}, vol. 244, pp. 229-244, 2015.



\bibitem{WLY2021}
Q. Wang, J. Liu and M. Ying,``Equivalence checking of quantum finite-state machines," \emph{Journal of Computer and System Sciences}, vol. 116, pp. 1-21, 2021.



\bibitem{TT2005} D. Thorsley and D. Teneketzis,``Diagnosability of Stochastic Discrete-Event Systems,"\emph{IEEE Transactions on Automatic Control}, vol. 50, no. 4, pp. 476-492, 2005.



\bibitem{KH2015}
C. Keroglou and C. N. Hadjicostis, ``Detectability in stochastic discrete event systems", \emph{ Systems \& Control Letters}, vol. 84, pp. 21-26, 2015.






\bibitem{LQXF2008} F. Liu and D. W. Qiu,``Safe diagnosability of stochastic discrete event systems,"\emph{IEEE Transactions on Automatic Control}, vol. 53, no. 5, pp. 1291-1296, 2008.







\end{thebibliography}
\end{document}